\documentclass[%
 aip,
 amsmath,amssymb,
preprint,%
]{revtex4-1}
\usepackage{mathrsfs,amsmath,amssymb,bm}
\usepackage{hyperref}
\usepackage{graphicx}
\usepackage{subfigure}
\usepackage{caption}
\usepackage{overpic}
\usepackage{epsfig}
\usepackage{xcolor}
\usepackage{indentfirst}
\usepackage{fancyhdr}
\usepackage{comment}
\usepackage{makecell, rotating}
\usepackage{algorithm} 
\usepackage{algorithmic} 
\usepackage{amsthm}
\newtheorem{lemma}{Lemma}
\theoremstyle{remark}
\newtheorem{remark}{Remark}

\begin{document}

%\setpagewiselinenumbers
%\modulolinenumbers[1]
%\linenumbers

\title{ Numerical Methods for Quasicrystals }


\author{Kai Jiang}
 \email{kaijiang@xtu.edu.cn}
 \affiliation{LMAM, CAPT and School of Mathematical Sciences, Peking University, Beijing 100871, P.R. China}
 \affiliation{Hunan Key Laboratory for Computation and
 Simulation in Science and Engineering, Xiangtan University,
 Xiangtan 411105, P.R. China}

 \author{Pingwen Zhang}
\thanks{}
 \email{pzhang@pku.edu.cn}
 \affiliation{LMAM, CAPT and School of Mathematical Sciences, Peking University, Beijing 100871, P.R. China}
\homepage{http://www.math.pku.edu.cn/pzhang}

\date{\today}

\begin{abstract}
Quasicrystals are one kind of space-filling structures.
The traditional crystalline approximant method utilizes periodic
structures to approximate quasicrystals. The errors of this
approach come from two parts: the numerical discretization, and
the approximate error of Simultaneous Diophantine Approximation
which also determines the size of the domain necessary for
accurate solution.
As the approximate error decreases, the computational
complexity grows rapidly, and moreover, the approximate error
always exits unless the computational region is the full space.
In this work we focus on the development of numerical method to
compute quasicrystals with high accuracy.  
With the help of higher-dimensional reciprocal space, a new
projection method is developed to compute quasicrystals. 
The approach enables us to calculate quasicrystals rather than
crystalline approximants.  
Compared with the crystalline approximant method,
the projection method overcomes the restrictions of the
Simultaneous Diophantine Approximation, and
can also use periodic boundary conditions conveniently.
Meanwhile, the proposed method efficiently reduces the
computational complexity through implementing in a unit cell and
using pseudospectral method.
For illustrative purpose we work with the Lifshitz-Petrich model,
though our present algorithm will apply to more general systems
including quasicrystals. We find that the projection method
can maintain the rotational symmetry accurately.
More significantly, the algorithm can calculate the free energy
density to high precision. 
\end{abstract}

%\begin{keyword}
%%% keywords here, in the form: keyword \sep keyword
%Quasicrystals;
%Crystalline approximant method;
%Projection method;  
%Projective matrix;  
%Lifshitz-Petrich model;
%Dodecagonal quasicrystal
%%Phase diagram
%\end{keyword}

\maketitle

%\section{}
%\label{}

\section{Introduction}
\label{sec:introduction}

As early as the 1890s, the periodic structures (crystals) in
three dimensions were determined by 230 space groups 
based on periodicity, and then the classical
crystallography was completed, in which the allowed rotational
symmetry is only 1-, 2-, 3-, 4-, 6-fold symmetry. Both the structure
determination and the study of physical properties are based on
the periodicity which allows the study to be simplified to a
unit cell. However, in the 1980s, a forbidden 5-fold symmetry
electron diffraction pattern was discovered by Shechtman
et al.\,\cite{shechtman1984metallic} in a rapid cooled
Al-Mn alloy. Later, the term ``quasicrystals'' appeared for the
first time to describe the non-conventional ordered
structures\,\cite{levine1984quasicrystals}. In an idealized description, 
quasicrystals have quasiperiodic, rather than periodic,
translational order with non-crystallographic symmetry. 
In the early theory of tilings, the discovery of aperiodic tiling
with 5-fold symmetry of the plane by Penrose\,\cite{penrose1974role,
gardner1997extraordinary} showed that such well-ordered systems
were mathematically possible.
Since the original discovery, hundreds of quasicrystals have been
reported and confirmed in metallic
alloys with 5-, 8-, 10-, 12- fold orientational
symmetry\,\cite{steurer2004twenty, tsai2008icosahedral}.
Two decades after the first discovery of quasicrystals in metallic alloys,
several soft quasicrystals have been found in soft matter
systems\,\cite{zeng2004supramolecular, hayashida2007polymeric,
talapin2009quasicrystalline, fischer2011colloidal}. 
The building blocks of the solid-state quasicrystals are the
atoms or small molecular on the atomic scale, whilst the building
blocks of soft quasicrystals are on a much larger scale of ten to
hundreds of nanometers. Therefore the continuous density
distribution are more appropriate for studying soft
quasicrystals. Accordingly the coarse-grained free energy
of density functions have been widely applied to treat phases and
phase transitions, especially for soft matter
systems\,\cite{gompper1994self, dotera2007mean,
barkan2011stability}.

Quasicrystals are one kind of space-filling structures.
Two kinds of theoretical approaches have been developed to
study quasicrystals, or more generally, aperiodic
crystals. The first method is an approximate approach which
studies crystalline approximants rather than quasicrystals.
The crystalline approximants are periodic structures in which the
arrangements of lattices closely approximate
the local structures in quasicrystals. Many crystalline 
approximants related to quasicrystals have also been
discovered\,\cite{goldman1993quasicrystals}.
These approximants may play an important role in 
describing the local structures of quasicrystals, their
formation, stability, and physical properties.
The second approach is a direct method to study the quasicrystals or
aperiodic crystals in the hyperspace, called the higher-dimensional
approach. From this approach, a quasiperiodic structure may be viewed
as a periodic structure by extending it into a higher-dimensional
space. Its symmetries can be expressed in terms of the
conventional point groups and space groups of higher-dimensional
periodic crystals\,\cite{janssen2007aperiodic,
steurer2009crystallography}.
In the higher-dimensional description, quasiperiodic structures
result from irrational physical-space cuts of appropriate
periodic hypercrystal structures. It is so-called cut-and-project
method. The higher-dimensional approach
reveals the hidden structural correlations. 
To implement the higher-dimensional approach in
the direct space one must know the discrete lattice 
arrangements of higher-dimensional periodic structure. 
The embedded
spaces of $d$-dimensional quasiperiodic structures are
abstract spaces whose dimensions are more than three.
The dimensions of the embedded space 
are dependent on the symmetry of the
quasicrystal ($d>1$)\,\cite{steurer2009crystallography,
hiller1985crystallographic}. For example, the quasicrystals
with 5-, 8-, 10-, and 12-fold symmetry need to be embedded into
four-dimensional space. While for the quasiperiodic structures
with 7-, 9-, 18-fold symmetry, the dimension of the embedding spaces
increases to six. For other symmetries, embedding spaces with
even higher dimension will be needed. 
Therefore it makes the method difficult for implementing. 
Another disadvantage of applying the higher-dimensional approach
in direct space is that in practical problems one is required to
compute the continuous quasiperiodic distribution rather than the
discrete quasiperiodic lattice.
A convenient way to describe the quasiperiodic
structures is in the higher-dimensional reciprocal
space\,\cite{steurer2009crystallography}. 
By redefining point-group symmetry and space-group
symmetry in terms of gauge functions, a broader Fourier-space
crystallography\,\cite{rabson1991space, mermin1992space,
drager1996superspace} has also been developed to describe
quasiperiodic structures. 

The traditional idea for treating quasicrystals is using a
periodic structure to approximate the quasiperiodic structure.
The method has been applied to molecular dynamics
simulations\,\cite{dzugutov1993formation, quandt1999formation,
skibinsky1999quasicrystals, engel2007self}, Monte Carlo
simulations\,\cite{dotera2006dodecagonal}, numerical
discretization methods\,\cite{lifshitz1997theoretical}.
In other words, those methods compute crystalline approximants
rather than quasicrystals.
A natural expectation is that the obtained approximants should
approximate quasicrystals as the computational box goes to
infinity. The advantage of the approximate method is that
using the periodic boundary conditions is convenient.
However, the approach can not
obtain quasicrystals exactly unless the computational box is the
full space because of 
the restriction of the Simultaneous Diophantine Approximation
(SDA)\,\cite{davenport1946simultaneous}. 
We will further explain this point in
Sec.\,\ref{sec:numericlMethods}.
In this paper, we focus on the development of numerical methods
for generating quasicrystals to high precision,
rather than crystalline approximants.
The proposed approach is based on
the observation that the Fourier spectrum of a $d$-dimensional
quasicrystal, consists of $\delta$
peaks on a $\mathbb{Z}$-module, $\mathbf{k}=\sum_{i=1}^n h_i
\mathbf{p}_i\in\mathbb{R}^d$, $h_i\in\mathbb{Z}$, of rank $n$
($n>d$) with basis vectors
$\mathbf{p}_i$\,\cite{janssen2007aperiodic,
steurer2009crystallography}.
Therefore, an natural idea is directly computing the Fourier
spectrum of quasicrystals instead of using periodic cells
to compute crystalline approximants in real space.

\section{Lifshitz-Petrich model}

Although our proposed  method is applicable to any model including
quasicrystals. For illustrative purpose we utilize the Lifshitz-Petrich
model\,\cite{lifshitz1997theoretical} to demonstrate our
algorithm. Before we go further,
a short introduction to the Lifshitz-Petrich model is necessary. 
Lifshitz-Petrich model is a coarse-grained mean-field theory. It 
is specially appropriate for studying the phase behaviour of
soft matters involving quasicrystals\,\cite{dotera2007mean,
barkan2011stability, lifshitz2007soft}. In particular, the Lifshitz-Petrich free
free energy density functional is
\begin{align}
	F[\phi(\mathbf{r})] = \frac{1}{V}\int
	d\mathbf{r}\,\Big\{\frac{c}{2}[(\nabla^2+1)(\nabla^2+q^2)\phi]^2
	-\frac{\varepsilon}{2}\phi^2-\frac{\alpha}{3}\phi^3+\frac{1}{4}\phi^4\Big\}.
	\label{eqn:LP}
\end{align}
In Eqn.\,(\ref{eqn:LP}), $\phi(\mathbf{r})$ is the order parameter.
$V$ is the system volume. $q$ is an irrational number depending
on the symmetry. $\varepsilon$ is the reduced temperature. $c>0$
is an energy penalty to ensure that the principle reciprocal
vectors of structures is located on $|\mathbf{k}|=1$ and
$|\mathbf{k}|=q$. $\alpha>0$ is a phenomenological parameter. 
For quasicrystals the system volume $V$ should go to
infinity since quasicrystals are the space-filling structures
without periodicity.  
The most significant feature of the Lifshitz-Petrich model is
the existence of two characteristic length scales, $1$ and $q$,
which is a necessary condition to stabilize the
quasicrystals\,\cite{lifshitz1997theoretical, jiang2013theory}.
Therefore it is a suitable
model to demonstrate our proposed approach.

Theoretically, the ordered patterns, including periodic and
quasiperiodic, are corresponding to local minima of the free
energy functional of the system with respect to order parameter
$\phi$. Accordingly the order parameter $\phi^*$ is the minimum
of the free energy density functional, which means 
\begin{align}
	\frac{\delta F}{\delta \phi(\mathbf{r})}\bigg|_{\phi^*} = 0.
	\label{}
\end{align}
In order to find the equilibrium state,
we introduce a relaxational dynamical  
to minimize the Lifshitz-Petrich energy functional
which yields 
\begin{align}
	\frac{\partial\phi}{\partial t}=-\frac{\delta F}{\delta\phi} =
	-c(\nabla^2+1)^2(\nabla^2+q^2)^2\phi+\varepsilon\phi+\alpha\phi^2-\phi^3.
	\label{eqn:AC}
\end{align}
We choose a semi-implicit scheme to solve the
dynamical equation\,(\ref{eqn:AC}):
\begin{align}
	\frac{1}{\Delta t}(\phi_{t+\Delta t}-\phi_{t}) = \varepsilon
	\phi_{t} -
	c(\nabla^2+1)^2(\nabla^2+q^2)^2\phi_{t+\Delta t}
	+\alpha(\phi^2)_t - (\phi^3)_t.
	\label{eqn:semiIter}
\end{align}
$\Delta t$ is the time step size.
The specific implementation of the semi-implicit method for
different numerical methods will be discussed in
Sec.\,\ref{sec:numericlMethods}.
The current paper is devoted to
develop a numerical method to calculate quasicrystals rather
than crystalline approximants. 

\section{Numerical Methods}
\label{sec:numericlMethods}

\subsection{Crystalline Approximant Method (CAM)}
\label{subsec:cam}

\subsubsection{Method Description of CAM}
\label{subsubsec:cam}

Numerical methods designed for periodic structures
have been used to study quasicrystals
approximately\,\cite{lifshitz1997theoretical}. Here we call
this method the ``crystalline approximant method (CAM)''.
In order to describe this method, a brief introduction of
numerical methods for treating periodic structures is necessary,
more details can be found in \onlinecite{zhang2008efficient}.
For any $d$-dimensional periodic function $f(\mathbf{r})$,
$\mathbf{r}\in\mathbb{R}^d$, 
the repeated structural unit is called a unit cell. A primitive
unit cell, described by $d$ vectors $\mathbf{e}_1$,
$\mathbf{e}_2$, \dots, $\mathbf{e}_d$, has the smallest possible
volume. The Bravais lattice vector is then defined by 
\begin{align}
	\mathbf{R} = l_1\mathbf{e}_1+l_2\mathbf{e}_2+\dots +
	l_d\mathbf{e}_d,
	\label{}
\end{align}
where $\mathbf{l}=(l_1, l_2, \dots, l_d)$ is a $d$-dimensional
vector with components $l_i\in\mathbb{Z}$. For any $\mathbf{R}$
in the Bravais lattice, the structure is invariant under a
lattice translation, i.e.,
$f(\mathbf{r}+\mathbf{R})=f(\mathbf{r})$.  Given the primitive
vectors $(\mathbf{e}_1, \mathbf{e}_2,\dots,\mathbf{e}_d)$, the
primitive reciprocal vectors $(\mathbf{e}^*_1,
\mathbf{e}^*_2,\dots,\mathbf{e}^*_d)$ satisfy the equation
\begin{align}
\mathbf{e}_i\cdot\mathbf{e}^*_j = 2\pi\delta_{ij}.
	\label{eqn:dualRelation}
\end{align}
The reciprocal lattice vector is then specified by
\begin{align}
	\mathbf{k} = k_1\mathbf{e}^*_1+k_2\mathbf{e}^*_2+\dots +
	k_d\mathbf{e}^*_d,
	\label{eqn:recipvectorPS}
\end{align}
where $k_i\in\mathbb{Z}$.
One of the most important properties of the 
reciprocal lattices is that plane waves
$\{e^{i\mathbf{k}\cdot\mathbf{r}}\}$ form a set of basis functions
for any function with periodicity of the lattice.
The periodic function $f(\mathbf{r})$ on the Bravais lattice can
be expanded as 
\begin{align}
	f(\mathbf{r}) = \sum_{\mathbf{k}}
	\hat{f}(\mathbf{k})e^{i \mathbf{k}\cdot\mathbf{r}}.
	\label{eqn:periodFourier}
\end{align}
For periodic structures, the reciprocal lattice vectors have two
important features: $\mathbf{e}^*_1$, $\mathbf{e}^*_2$, \dots,
$\mathbf{e}^*_d$ are linearly independent; and
$k_i\in\mathbb{Z}$, $i = 1,2,\dots,d$.  In the direct space, the
lattice vectors have the same properties.

The main idea of CAM is to use periodic structures to
approximate quasicrystals. For a $d$-dimensional quasicrystals, 
its reciprocal lattice vectors $\mathbf{k}$ 
can be expressed by $d$ linearly independent reciprocal
vectors, $\mathbf{e}^*_1, \mathbf{e}^*_2, \dots, \mathbf{e}^*_d$,
\begin{align}
	\mathbf{k} = p_1\mathbf{e}^*_1+p_2\mathbf{e}^*_2+\dots +
	p_d\mathbf{e}^*_d.
	\label{eqn:recipvectorQC}
\end{align}
It is important to note that the $p_i\in\mathbb{R}$, usually
including irrational numbers, that is,
$\mathbf{k}$ can not be represented by linear combinations of
$\mathbf{e}^*_i$ with integer-valued coefficients.
However, the quasiperiodic function $\phi(\mathbf{r})$ can be
expanded as in the following form
\begin{align}
	\phi(\mathbf{r}) = \sum_{\mathbf{k}}
	\hat{\phi}(\mathbf{k})e^{i \cdot(L\mathbf{k})\cdot\mathbf{r}/L}, ~~~
	\mathbf{r}\in[0,2\pi L)^d.
	\label{eqn:pam:fourier}
\end{align}
If there exists a rational number $L$ such that $L p_i\in
\mathbb{Z}$
or $L p_i$ can be made arbitrarily close to a series of integers,
for all $i=1,2,\dots,d$ and $\mathbf{k}$, which means
\begin{align}
	|L\cdot(p_1, p_2, \dots, p_d)-([Lp_1], [Lp_2], \dots,
	[Lp_d])|_{l^{\infty}} \rightarrow 0, ~~\mbox{for all} ~~
	\mathbf{k}, 
	\label{eqn:sda}
\end{align}
where $[\cdot]$ rounds the number $\cdot$ to the nearest integer.  
Then numerical methods designed for
periodic structures can be used to treat quasicrystals.
Without loss of generality, we can always choose
$(1, 0, \dots, 0)$ as one of the primitive
reciprocal vectors. Therefore, the rational number $L$ becomes an
integer. The problem of determining $L$ is a
well-known problem, which deals with the approximation of real
numbers by rational numbers or integers, called the Simultaneous
Diophantine Approximation (SDA) in number
theory\,\cite{davenport1946simultaneous}. In view of numerical
computability, the integer $L$ should be as small as possible.
For simplicity, we denote the approximate error of SDA as $E_{SDA}$.

From the above description, the condition of
SDA must be satisfied for the implementation of CAM. Therefore,
there are two sources of errors of the CAM, from the
discretization, and from the Diophantine approximation.
The integer $L$ is dependent on these irrational numbers
$p_i$ due to rotational symmetry and the desired precision of the
approximation. In the subsequent numerical examples
(in Sec.\,\ref{subsec:complexity}), we will find
that $L$ increases very quickly as the desired precision
becomes higher. This greatly increases the computational cost.
For CAM, the computational
box should satisfy the condition (\ref{eqn:sda}). Therefore,
the edge length $D_\mathbf{k}$ of computational box is close to 
$D_\mathbf{k}=L\times
2\pi$, or $D=n \cdot L\times 2\pi$, $n \in \mathbb{N}^+$.

\subsubsection{An Example of 12-fold Rotational Symmetry}
\label{subsubsec:cam:example}
\begin{figure}[!hbtp]
	\centering    
		\includegraphics[scale=0.5]{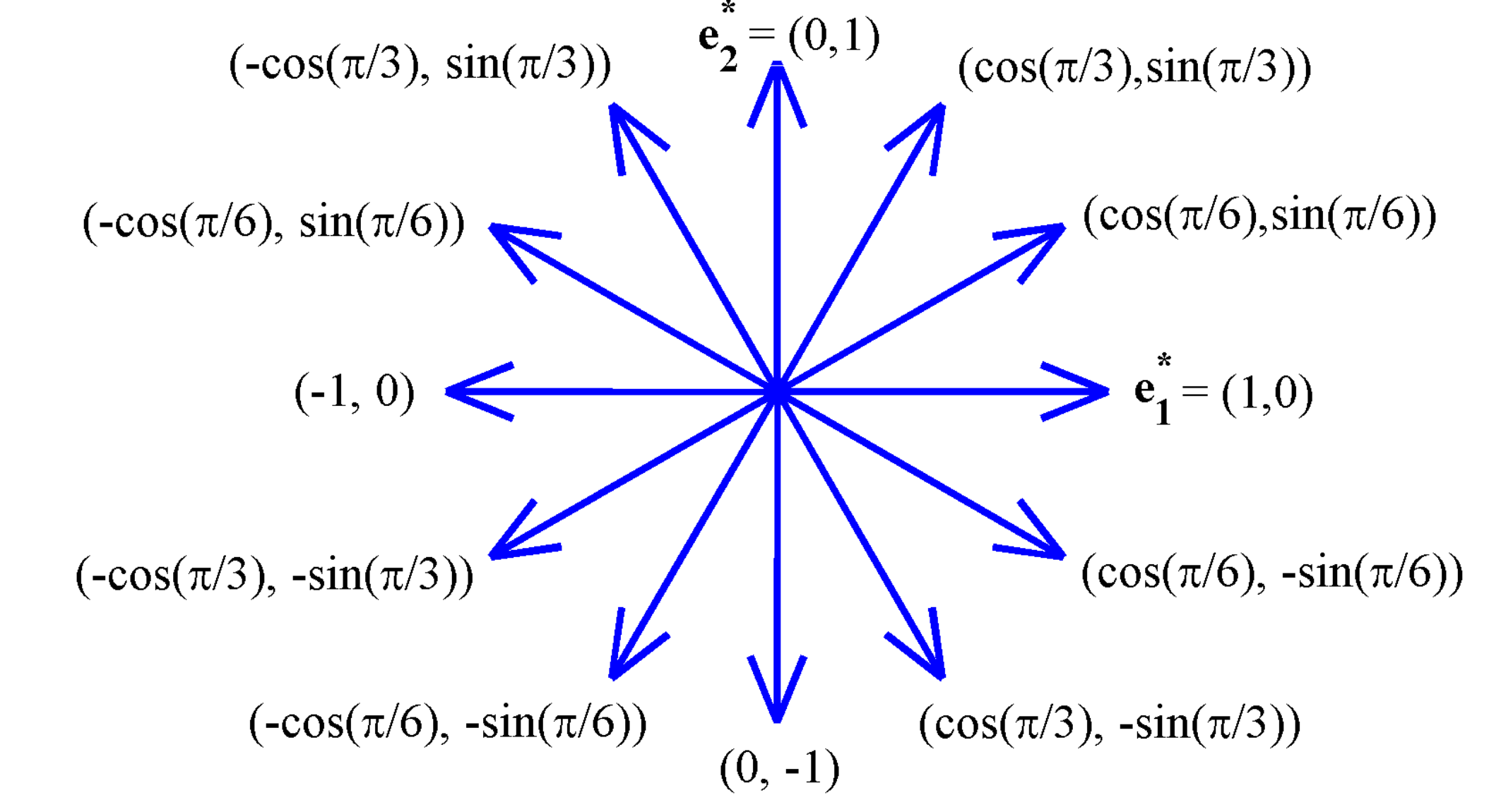}
    \caption{\label{fig:pamscheme}
	Sets of reciprocal lattice vectors $\{\mathbf{k}\}$ 
	of 12-fold rotational symmetry in 2-dimensional space.
	}
\end{figure}
In order to make the above method clearer,
we take an example of 12-fold rotational symmetry
in two-dimensions to demonstrate this point. A schematic plot in the
reciprocal space is shown in Fig.\,\ref{fig:pamscheme}.
Let two noncollinear vectors $\mathbf{e}^*_1=(1,0)$ 
and $\mathbf{e}^*_2=(0,1)$ comprise the primitive 
reciprocal vectors, other reciprocal vectors of the 12-fold
case are represented as linear combinations of
$\mathbf{e}^*_1$ and $\mathbf{e}^*_2$ with real number
coefficients, as shown in Fig.\,\ref{fig:pamscheme}. 
If other two noncollinear reciprocal vectors are chosen as 
primitive reciprocal vectors, similar results will emerge as
well. The distinct nonzero coefficients for these $12$ vectors
are $1$, $1/2$, and $\sqrt{3}/2$.
When applying CAM to this example of 12-fold vectors,
we need an integer $L$ such that
\begin{align}
	\Big|L\cdot \Big(1, \frac{\sqrt{3}}{2}, \frac{1}{2}\Big) - \Big(L,
	\Big[\frac{\sqrt{3}}{2}L\Big],
	\Big[\frac{1}{2}L\Big]\Big)\Big|_{l^{\infty}} \rightarrow 0.
	\label{eqn:sda:DDQC}
\end{align}

\subsubsection{Applying CAM to Lifshitz-Petrich Method}
\label{subsubsec:cam:lp}

In CAM, we indeed use crystalline approximants to approximate
quasicrystals. By substituting Eqn\,(\ref{eqn:pam:fourier})
into\,(\ref{eqn:LP}), the Lifshitz-Petrich free energy density
functional becomes  
\begin{align}
	\nonumber
	F[\phi] =
	\frac{1}{2}&\sum_{\mathbf{k}_1+\mathbf{k}_2=0}[c(1-|\mathbf{k}|^2)^2(q^2-|\mathbf{k}|^2)^2-\varepsilon]\hat{\phi}(\mathbf{k}_1)\,\hat{\phi}(\mathbf{k}_2)
	\\
	&-\frac{\alpha}{3}\sum_{\mathbf{k}_1+\mathbf{k}_2+\mathbf{k}_3=0}\hat{\phi}(\mathbf{k}_1)\,\hat{\phi}(\mathbf{k}_2)\,\hat{\phi}(\mathbf{k}_3)
	+\frac{1}{4}\sum_{\mathbf{k}_1+\mathbf{k}_2+\mathbf{k}_3+\mathbf{k}_4=0}\hat{\phi}(\mathbf{k}_1)\,\hat{\phi}(\mathbf{k}_2)\,\hat{\phi}(\mathbf{k}_3)\,\hat{\phi}(\mathbf{k}_4).
	\label{eqn:LPperiodFour}
\end{align}
Then we use the semi-implicit scheme (\ref{eqn:semiIter}) 
to minimize the free energy density functional. 
For CAM, the semi-implicit method becomes 
\begin{align}
	\left(\frac{1}{\Delta
	t}+c(1-\mathbf{k^2})^2(q^2-\mathbf{k^2})^2\right)\hat{\phi}_{t+\Delta
	t}(\mathbf{k})= \left(\frac{1}{\Delta
	t}+\varepsilon\right)\hat{\phi}_t(\mathbf{k}) + \alpha
	(\phi^2_t)(\mathbf{k})-(\phi^3_t)(\mathbf{k}),
	\label{eqn:semiPAM}
\end{align}
where $ (\phi^m_t)(\mathbf{k})=\int
d\mathbf{r}\,\phi^m(\mathbf{r})e^{-i\mathbf{k}\mathbf{r}}$,
$m =2,3$. The right terms of the dynamics can be efficiently
calculated by the pseudospectral
method\,\cite{zhang2008efficient}. The Laplacian terms are
computed in $d$-dimensional reciprocal space easily, while the
convolutions can be calculated efficiently by Fast-Fourier
transformation (FFT). Then the computational complexity is
$O(N\log N)$ at each time step, $N$ is the number of degrees of freedom.

\subsection{Projection Method (PM)}
\label{subsec:pm}

\subsubsection{Method Description of PM}
\label{subsubsec:pm}

As mentioned above, CAM computes the crystalline approximants.
However, the approximation method can not evaluate the free
energy density exactly because of the approximate error of SDA.
Therefore it is necessary to design a new numerical method
that improves the calculation of quasicrystals by avoiding the
approximate error of SDA.
From the higher-dimensional
description\,\cite{steurer2009crystallography}, there exists an
equivalent $n$-dimensional representation of a $d$-dimensional
quasicrystal ($n>d$). 
The reciprocal vectors $\mathbf{k}\in\mathbb{R}^d$ of a $d$-dimensional
quasicrystal are 
\begin{align}
	\mathbf{k}= h_1\mathbf{p}^*_1 + h_2\mathbf{p}^*_2 + \dots
	+h_n\mathbf{p}^*_n,
	~~~h_i\in\mathbb{Z}, 
	\label{eqn:hdm:d}
\end{align}
with vectors $\mathbf{p}^*_i\in\mathbb{R}^d$ of rank $n$.
In order to develop the new numerical method, 
we firstly redefine the selected reciprocal vectors
$\mathbf{p}^*_i$ with $n$ components.
Then the quasicrystal can be represented by these redefined
basic reciprocal vectors with integral coefficients in
$n$-dimensional space. This representation has the advantage,
for our purpose, that the $d$-dimensional quasicrystal is
periodic in $n$-dimensions. 
Based on the redefined reciprocal vectors, we
introduce the $n$-dimensional reciprocal lattice. Assume that
$n$-dimensional vectors, 
$\mathbf{b}^*_1, \mathbf{b}^*_2, \dots, \mathbf{b}^*_n$,
are the primitive reciprocal vectors of a $1$st Brillouin zone
in $n$-dimensional reciprocal space, the reciprocal vector of
$n$-dimensional periodic structure can be expressed as  
\begin{align}
	\mathbf{H} = h_1\mathbf{b}^*_1+h_2\mathbf{b}^*_2+\dots +
	h_n\mathbf{b}^*_n,
	\label{eqn:hdm:n}
\end{align}
where the coefficient $h_i\in\mathbb{Z}$, and
$\mathbf{H}\in\mathbb{R}^n$.
The correspondingly primitive
vectors of the direct space $\mathbf{b}_i$ satisfy the dual
relationship (\ref{eqn:dualRelation}).
We want to use the $n$-dimensional reciprocal vector $\mathbf{H}$ 
to represent the $d$-dimensional quasicrystal.
Then we can solve a $d$-dimensional quasicrystal as a periodic
structure in $n$-dimensional space.
The key point of implementing the above idea is to provide
an operator to project the $n$-dimensional structure into
$d$-dimensional space.

In order to solve the problem, we propose a novel Fourier
expansion for the $d$-dimensional quasiperiodic function 
\begin{align}
	g(\mathbf{r}) =
	\sum_{\mathbf{H}}\hat{g}(\mathbf{H})e^{i[(\mathcal{S}\cdot\mathbf{H})^{T}\cdot\mathbf{r}]}.
	\label{eqn:qcFourier}
\end{align}
In the expression, $\mathbf{r}\in\mathbb{R}^d$,
$\mathbf{H}\in\mathbb{R}^n$,
and $\mathcal{S}$ is the \textit{projective
matrix} which connects the $d$-dimensional physical 
space with the $n$-dimensional reciprocal space. 
We note that the two representations (\ref{eqn:hdm:d}) and
(\ref{eqn:hdm:n}) can be used to describe the same quasicrystal. The
reciprocal vectors of a $d$-dimensional quasicrystal can be
represented by $d$-dimensional reciprocal vectors
$\mathbf{p}^*_i$ with integral coefficients, and also the
integral combinations of extended $n$-dimensional primitive reciprocal
vectors $\mathbf{b}_i^*$. Therefore, we can obtain the projective
matrix $\mathcal{S}$ through projecting the $n$-dimensional
reciprocal vectors $\mathbf{b}_i$ into $d$-dimensional
reciprocal space, i.e., $\mathbf{p}_i = (\mathcal{S}\cdot\mathbf{b})_i$,
$i=1,2,\dots,n$. The $j$-th component $p^*_{ij}$ of the projected
reciprocal vector $\mathbf{p}_i^*$ can be expressed by $\mathbf{b}^*_i$ 
\begin{align}
	p^*_{ij} = \sum_{m=1}^n s_{jm} b^*_{im}, ~~~ j = 1,2,\dots,d, 
	\label{}
\end{align}
where $b^*_{im}$ is the $m$-th component of the $i$-th $n$-dimensional
reciprocal vector $\mathbf{b}^*_i$. These coefficients
$s_{jm}\in\mathbb{R}$, $j=1,2,\dots,d$, $m=1,2,\dots,n$, form
the $d\times n$-order nonzero projective matrix $\mathcal{S}$
which reflects the symmetry of quasicrystals.
For the least computational expense, the dimension
$n$ of the extended space must be the smallest determined by the
order of the elements of the symmetric
group\,\cite{steurer2009crystallography,
hiller1985crystallographic}. For example, 5-, 8-, 10-, and
12-fold symmetric quasicrystals, the dimension of embedded space
is four. However, 7-, 9-, and 18-fold symmetric quasicrystals
must be restricted to six-dimensional space.
The projective matrix $\mathcal{S}$ is not unique, which is
determined by the selection of reciprocal vectors $\mathbf{b}^*_i$. 
The represented coefficients of reciprocal vectors are dependent
on the selection of primitive reciprocal vectors,
however, the reciprocal vectors of a quasicrystal are unique. 
Therefore, the selection of basic reciprocal vectors
$\mathbf{b}_i^*$ as well as the projective matrix $\mathcal{S}$
is irrelevant to the quasicrystal.  Considering a periodic
structure, the projective matrix $\mathcal{S}$ degenerates to a
$d\times d$-order unit matrix.

Furthermore, PM can be extended to calculate
one-dimensional quasicrystals even though the notion of
rotational symmetry does not exist in one dimension.
If an energy functional has one-dimensional
quasiperiodic structures with different incommensurate scales,
for example, 
\begin{align}
	g(x) = C_0 \sin(x) + C_1 \sin(q_1 x) + C_2\sin(q_2 x) +
	\cdots,
	\label{eqn:hdm:1d}
\end{align}
where the common multiples of $1, q_1, q_2, \dots$, are irrational
numbers in pairs, the projective matrix becomes a vector
\begin{align}
	\mathcal{S} = (1, q_1, q_2, \dots).
	\label{}
\end{align}
The present approach can be applied to calculate
one-dimensional quasicrystals.

In the following, a lemma is given to
indicate which variable should be computed in PM.
\begin{lemma}
	\label{lem}
	For a $d$-dimensional quasiperiodic function $g(\mathbf{r})$,
	under the expansion (\ref{eqn:qcFourier}), 
	we have 
\begin{align}
	\lim_{V\rightarrow\infty}\frac{1}{V}\int
	d\mathbf{r}\,g(\mathbf{r}) =
	\hat{g}(\mathbf{H})\Big|_{\mathbf{H}=0}.
	\label{eqn:pm:energy:I}
\end{align}
\end{lemma}

\begin{proof}
	Firstly, we note that 
\begin{align}
	\lim_{V\rightarrow\infty}\frac{1}{V}\int
	d\mathbf{r}\,\exp\Big\{(\mathcal{S}\cdot\mathbf{H})\cdot\mathbf{r}\Big\}
	= \delta(\mathcal{S}\cdot\mathbf{H}).
\end{align}
Therefore, 
\begin{align}
	\lim_{V\rightarrow\infty}\frac{1}{V}\int
	d\mathbf{r}\,g(\mathbf{r}) =
	\lim_{V\rightarrow\infty}\frac{1}{V}\int d\mathbf{r}\,
	\sum_{\mathbf{H}}\hat{g}(\mathbf{H})e^{i[(\mathcal{S}\cdot\mathbf{H})^{T}\cdot\mathbf{r}]}
	= 
	\hat{g}(\mathbf{H})\Big|_{\mathcal{S}\cdot\mathbf{H}=0}.
	\label{eqn:pm:energy:II}
\end{align}
Then we just need to prove that the expressions (\ref{eqn:pm:energy:I})
and (\ref{eqn:pm:energy:II}) are equivalent. From the definition
of (\ref{eqn:hdm:n}), $n$-dimensional reciprocal vector
$\mathbf{H}$ is the integer-valued combinations of linearly
independent primitive reciprocal vectors $\mathbf{b}^*_i$,
$i=1,2,\dots,n$. Because the projective matrix $\mathcal{S}$
is linear, the $d$-dimensional projected vector
$\mathcal{S}\cdot\mathbf{H}$ is also the integer-valued
combinations of the $d$-dimensional projected vectors
$\mathbf{p}^*_i=\mathcal{S}\cdot\mathbf{b}^*_i$. 
$\mathcal{S}\cdot\mathbf{H}=0$ is equivalent to the integral
coefficients $h_i$ in Eqn.\,(\ref{eqn:pm:energy:II}) to be zero.
It means that the $n$-dimensional reciprocal vector $\mathbf{H}=0$. 
\end{proof}
\begin{remark}
From Lemma 1, the Fourier coefficients
$\hat{g}(\mathbf{H})$ rather than
$\hat{g}(\mathcal{S}\cdot\mathbf{H})$ should be computed in PM.
\end{remark}

In the PM, a quasicrystal is computed in $n$-dimensional
reciprocal space as a periodic structure, then the
$n$-dimensional structure is projected into $d$-dimensional space
to obtain the $d$-dimensional quasicrystal through projective
matrix. Since different periodic phases have their own
periodicity, the appropriate computational box is important to
determine the final morphology of solutions, especially for
complex phases. For example, in diblock copolymer
systems\,\cite{jiang2013discovery}, the lamellae phase can be
obtained easily in any computational region, however, for complex 
gyroid pattern, the computational box should be close to its period.
For more complex quasicrystal, the computational box
should be estimated carefully before computing. 
We also note that a equilibrium periodic
structure is not only the minimum of a free energy density
$F$ with respect to order parameters, but also
with respect to the unit cell\,\cite{zhang2008efficient}.
Therefore, the unit cell in $n$-dimensional space,
$\mathcal{B}=[\mathbf{b}_1,\mathbf{b}_2, \dots, \mathbf{b}_n]$,
should satisfy $\partial F/\partial\mathcal{B} =
0$. We will give a estimation formula for the unit cell of
quasicrystals in the Lifshitz-Petrich model (see
Sec.\,\ref{subsubsec:pm:lp}). 

With the help of $n$-dimensional reciprocal space, the PM can
calculate the spectrum of quasicrystals directly.
In PM, the physical space variable
$\mathbf{r}$ always belongs to $d$-dimensional space.
Therefore an energy functional including
$d$-dimensional quasicrystals should not be up to 
describing $n$-dimensional structures.
The PM is able to compute quasicrystals rather than
crystalline approximants. The PM is also using the periodic
condition conveniently in higher-dimensional reciprocal space .	
Compared with CAM,
the most important advantage of PM is that
the method overcomes the restriction of SDA. Therefore PM can
compute the free energy density to high accuracy numerically. 
Meanwhile, the proposed method can implement in a $n$-dimensional
reciprocal unit cell to reduce computational cost. 

\subsubsection{An Example of 12-fold Rotational Symmetry}
\label{subsubsec:pm:example}

As an example, we take the 12-fold rotational symmetry in two
dimensions to illustrate the idea of PM. 
\begin{figure}[!hbtp]
	\centering    
		\includegraphics[scale=0.5]{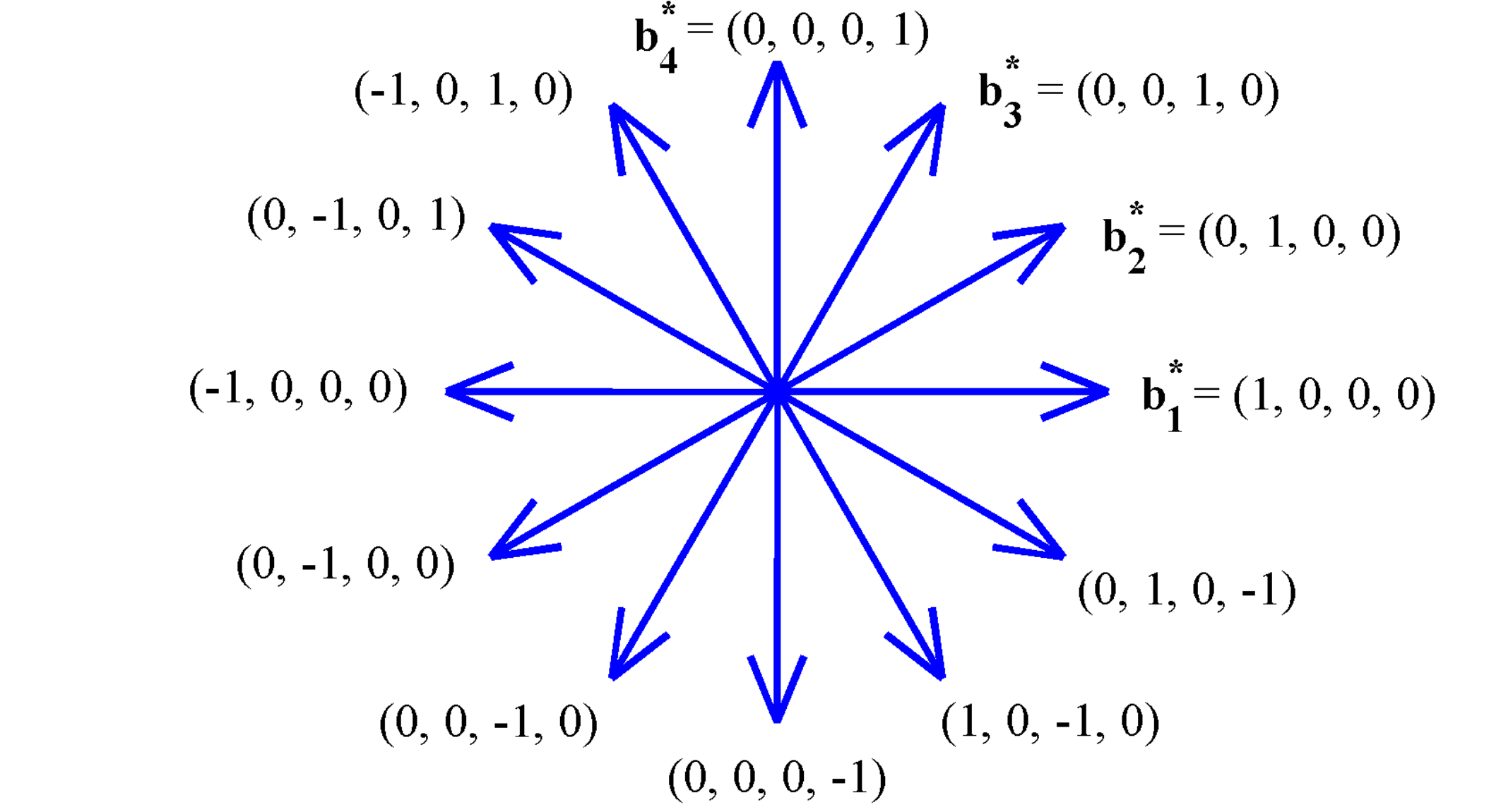}
    \caption{\label{fig:hdmscheme}
	Sets of reciprocal lattice vectors $\{\mathbf{H}\}$ 
	of the 12-fold rotational symmetry with 4 components.
	}
\end{figure}
As Sec.\,\ref{subsubsec:cam:example} discussed,
the $12$ reciprocal lattice vectors can not be represented
by two noncollinear vectors $\mathbf{e}^*_1$ and $\mathbf{e}^*_2$
with integral coefficients.
However, as Fig.\,\ref{fig:hdmscheme} shows, if the vectors of
$\mathbf{b}_1^*$, $\mathbf{b}_2^*$, $\mathbf{b}_3^*$, and
$\mathbf{b}_4^*$ with $4$ components are chosen as primitive
reciprocal vectors, other reciprocal lattice vectors can be
represented as integral combinations of these primitive vectors.
The selected $4$ vectors make up a basis in $4$-dimensional space.
The $2$-dimensional 12-fold example is related to a periodic
structure in $4$-dimensional reciprocal space. 
The projected vectors $\mathbf{p}^*_i =
\mathcal{S}\cdot\mathbf{b}_i^*$,
$i=1,2,3,4$ also can be expressed by $\mathbf{e}^*_1$,
$\mathbf{e}^*_2$ in $2$-dimensional space, i.e.,
\begin{align}
	p^*_{i1} = \cos\frac{(i-1)\pi}{6}e^*_{11}, ~~~~ 
	p^*_{i2} = \sin\frac{(i-1)\pi}{6}e^*_{22}.
	\label{}
\end{align}
Therefore, the projective matrix is
\begin{equation}
	\mathcal{S} =\left(
	\begin{array}{cccc}
		1 & \cos(\pi/6) & \cos(\pi/3) & 0 \\
		0 & \sin(\pi/6) & \sin(\pi/3) & 1 
	\end{array}
\right).
\label{eqn:DDQC:projMatrix}
\end{equation}

\subsubsection{Applying PM to Lifshitz-Petrich Method}
\label{subsubsec:pm:lp}
Applying PM to the Lifshitz-Petrich model in two dimensions, the
order parameter $\phi(x,y)$ is expanded as
\begin{align}
	\phi(x,y) =
	\sum_{\mathbf{H}}\hat{\phi}({\mathbf{H}})\exp\Big\{i(\mathcal{S}\cdot\mathbf{H})^{T}\cdot(x,y)^{T}\Big\},
	\label{}
\end{align}
where $\mathbf{H}$ is a $n$-dimensional vector, $\mathcal{S}$ is
a ${2\times n}$-order projective matrix. 
For quasicrystals, the system volume $V$ should go to infinity.
Based on Lemma 1, 
the Lifshitz-Petrich free energy density functional becomes  
\begin{align}
	\nonumber
	F[\phi(x,y)] &=
	\frac{1}{2}\sum_{\mathbf{H}_1+\mathbf{H}_2=0}\{c[1-(g_1^2+g_2^2)]^2[q^2-(g_1^2+g_2^2)]^2-\varepsilon\}\hat{\phi}(\mathbf{H}_1)\,\hat{\phi}(\mathbf{H}_2)
	\\
	&\hspace{-0.3cm} -\frac{\alpha}{3}\sum_{\mathbf{H}_1+\mathbf{H}_2+\mathbf{H}_3=0}\hat{\phi}(\mathbf{H}_1)\,\hat{\phi}(\mathbf{H}_2)\,\hat{\phi}(\mathbf{H}_3)
	+\frac{1}{4}\sum_{\mathbf{H}_1+\mathbf{H}_2+\mathbf{H}_3+\mathbf{H}_4=0}
	\hat{\phi}(\mathbf{H}_1)\,\hat{\phi}(\mathbf{H}_2)\,\hat{\phi}(\mathbf{H}_3)\,\hat{\phi}(\mathbf{H}_4),
	\label{eqn:LPqcFour}
\end{align}
where $g_1$ and $g_2$ are defined by
\begin{align}
	\mathcal{S}\cdot \mathbf{H} = \Big(
	\sum_{i=1}^n s_{1i}\sum_{j=1}^n h_j b_{ji},
	\sum_{i=1}^n s_{2i}\sum_{j=1}^n h_j b_{ji}
	\Big)^{T}\stackrel{\Delta}{=} (g_1, g_2)^{T},
	\label{eqn:hdm:laplace}
\end{align}
$h_j\in\mathbb{Z}$, $s_{1j}$ and $s_{2j}$ are the components of
the projective matrix $\mathcal{S}$, and $b_{ji}$ is the
$i$-th component of the primitive reciprocal vector $\mathbf{b}_j$.
The $n$-dimensional Fourier
coefficient $\hat{\phi}(\mathbf{H})$ can be solved by 
the semi-implicit method\,(\ref{eqn:semiIter}) 
\begin{align}
	\nonumber
&	\left(\frac{1}{\Delta
	t}+c\left(1-g_1^2-g_2^2\right)^2 \left(q^2-g_1^2-g_2^2\right)^2\right)\hat{\phi}_{t+ \Delta
	t}(\mathbf{H})
	\\
&	\hspace{3cm} = \left(\frac{1}{\Delta
	 t}+\varepsilon\right)\hat{\phi}_t(\mathbf{H}) + \alpha
	(\phi^2_t)(\mathbf{H})-(\phi^3_t)(\mathbf{H}),
	\label{eqn:semiPM}
\end{align}
where the quadratic term and the third term are 
\begin{align}
(\phi^2_t)(\mathbf{H})=\sum_{\mathbf{H_1}+\mathbf{H}_2=\mathbf{H}}
\hat{\phi}_t(\mathbf{H}_1)\,\hat{\phi}_t(\mathbf{H}_2),
~~~
(\phi^3_t)(\mathbf{H})=\sum_{\mathbf{H_1}+\mathbf{H}_2+\mathbf{H}_3=\mathbf{H}}
\hat{\phi}_t(\mathbf{H}_1)\,\hat{\phi}_t(\mathbf{H}_2)\,
\hat{\phi}_t(\mathbf{H}_3).
	\label{}
\end{align}
In Eqn.\,(\ref{eqn:semiPM}), the linear terms can be solved easily.
The nonlinear terms are $n$-dimensional convolutions in
reciprocal space. Directly computing will result in expansive
computational cost. However these convolutions are local point
multiplication in $n$-dimensional direct space.
Therefore we use the pseudospectral
method\,\cite{zhang2008efficient} to treat these terms by FFT. 
It should be emphasized that the PM is implemented in
$n$-dimensional reciprocal space instead of in $d$-dimensional
physical space. Computing these convolutions in 
$n$-dimensional direct space is to reduce computational complexity.

For PM, we give a method to estimate the $n$-dimensional unit cell
for the Lifshitz-Petrich model\,\cite{jiang2013discovery}. 
Without loss of generality, we can choose a proper coordinate
system such that $b_{ii}\neq 0$, $b_{ij}=0$, when $j>i$.
If $\mathcal{B}=[\mathbf{b}_1, \mathbf{b}_2, \dots,
\mathbf{b}_n]$ is a $n$-dimensional unit cell of a
$d$-dimensional quasicrystal. 
The first deviations of the free energy
functional with respect to $b_{ij}$ should be zero, where
$i=1,2,\dots n$, $j=1,2,\dots,i$, i.e.,
\begin{align}
	\sum_{\mathbf{H}}
	h_{j}(g_1 s_{1 i} + g_2 s_{2 i})
	[1-(g_1^2+g_2^2)][q^2-(g_1^2+g_2^2)]
	[1+q^2-2(g_1^2+g_2^2)]\cdot
	|\hat{\phi}(\mathbf{H})|^2 = 0.
	\label{eqn:box}
\end{align}
Then the unit cell can be obtained by solving the
Eqn.\,(\ref{eqn:box}). 
Since the Fourier coefficients $\phi_{\mathbf{H}}$ are
unknown beforehand. In practice, only primary reciprocal
vectors with equal Fourier coefficients because of symmetries are
considered in estimating the unit cell. Therefore the Fourier
coefficients $\hat\phi(\bm H)$ in \eqref{eqn:box} will be
cancelled and make no impression on evaluating the size of box.

\subsection{Computational Complexity}
\label{subsec:complexity}

In this section, we will give a general analysis on the
computational complexity of CAM and PM in solving the 
two-dimensional Lifshitz-Petrich model.
The computational complexity of CAM is dependent on the approximate
accuracy of SDA\,(\ref{eqn:sda}) and the numerical precision.
However, the PM overcomes the restriction of SDA, whose
computational complexity is only dependent on the numerical accuracy.
As described in Sec\,\ref{subsubsec:cam} and
Sec.\,\ref{subsubsec:pm}, both CAM and PM can use pseudospectral
method in computing the two-dimensional Lifshitz-Petrich model.
The number of basic functions used in the reciprocal space can be
equivalent to the number of discretized points in the direct
space. In order to guarantee the same numerical precision, 
we assume that the mesh step size is $\Delta x$ both for CAM
($2$ dimensions) and PM ($n$ dimensions) in the direct space.

In the Lifshitz-Petrich model, the CAM is
implemented in two dimensions. 
At a desired approximate error $E_{SDA}$ as well as the integer
$L$, the number of the plane waves is $N_{\mathbf{k}}\times N_{\mathbf{k}}$, 
\begin{align}
	N_{\mathbf{k}} = \left[\frac{D_{\mathbf{k}}}{\Delta
	x}\right], ~~D_{\mathbf{k}} = 2\pi \times L,
	\label{}
\end{align}
where $[\cdot]$ rounds the number $\cdot$ to the nearest integer.  
The computational complexity of CAM
\begin{align}
	O\left(N_{t,\mathbf{k}}\cdot 2N_{\mathbf{k}}^2 \log
	N_{\mathbf{k}}\right),
	\label{eqn:complexity:cam}
\end{align}
where $N_{t,\mathbf{k}}$ is the number of time iterations.

In PM, the $2$-dimensional quasicrystal is represented in
$n$-dimensional reciprocal space as a periodic structure.
The dimension of $n$ is dependent on the rotational symmetry
of a quasicrystal. 
Therefore the $2$-dimensional quasicrystal can be computed in
a $n$-dimensional unit cell.
We assume that the size of the $n$-dimensional
unit cell is $D_{\mathbf{H}}$ in the direct space. 
The number of the plane-wave functions is $N_{\mathbf{H}}^n$, with 
\begin{align}
	N_{\mathbf{H}} = \left[\frac{D_{\mathbf{H}}}{\Delta
	x}\right].
	\label{}
\end{align}
As discussed in Sec.\,\ref{subsubsec:pm},
the pseudospectral method in PM is used to solve
the dynamical equation\,(\ref{eqn:semiPM}) in $n$-dimensional
reciprocal space with the help of FFT.
Therefore the computational complexity of PM is 
\begin{align}
	O\left(N_{t,\mathbf{H}}\cdot n N_{\mathbf{H}}^n \log
	N_{\mathbf{H}}\right),
	\label{eqn:complexity:pm}
\end{align}
where $N_{t,\mathbf{H}}$ is the number of time iterations.

From the estimation formulas (\ref{eqn:complexity:cam}), the
computational complexity of CAM is a function of the integer
$L$, which is dependent upon the approximate error of SDA. 
In next section, we will find that $D_{\mathbf{k}}
\gg D_{\mathbf{H}}$ as the approximate error of SDA decreases.
The computational complexity of CAM may be larger than
that of PM in a situation where higher numerical accuracy is
required. More importantly, PM can compute quasicrystals and
their free energy density to high accuracy without any
approximate error of SDA. On the contrary, CAM always has the approximate
error of SDA unless the computational box goes to infinity which
results in unaccepted computational cost. 

\section{Numerical Results and Discussion}
\label{sec:numericalResults}

We will demonstrate the behavior of the two numerical methods,
CAM and PM, based on the two-dimensional Lifshitz-Petrich
model. In previous research\,\cite{lifshitz1997theoretical}, 
the Lifshitz-Petrich model showed that if
$q$ is chosen around $2\cos(\pi/12)$ one can obtain a
2-dimensional quasicrystal with dodecagonal (12-fold)
symmetry. The earlier work found that no choice of $q$ yields
globally stable octagonal or decagonal symmetric pattern. 
However, recent study\,\cite{jiang2013theory} manifests that the
decagonal quasicrystal is stable in the Lifshitz-Petrich model.
In this work we just consider the dodecagonal symmetric structure.

\subsection{Computational Complexity of Dodecagonal Symmetric
Structure}
\label{subsec:result:complexity}

For the dodecagonal quasicrystal (DDQC), the initial reciprocal
vectors at which the Fourier coefficients are nonzero, are shown
in Fig.\,\ref{fig:12fold:initial}. They contain two 12-fold stars
of wave vectors, one on $|\mathbf{k}|=1$ and other on
$|\mathbf{k}|=q=2\cos(\pi/12)$\,\cite{lifshitz1997theoretical}. 
\begin{figure}[!htp]
	\centering
		\includegraphics[scale=0.5]{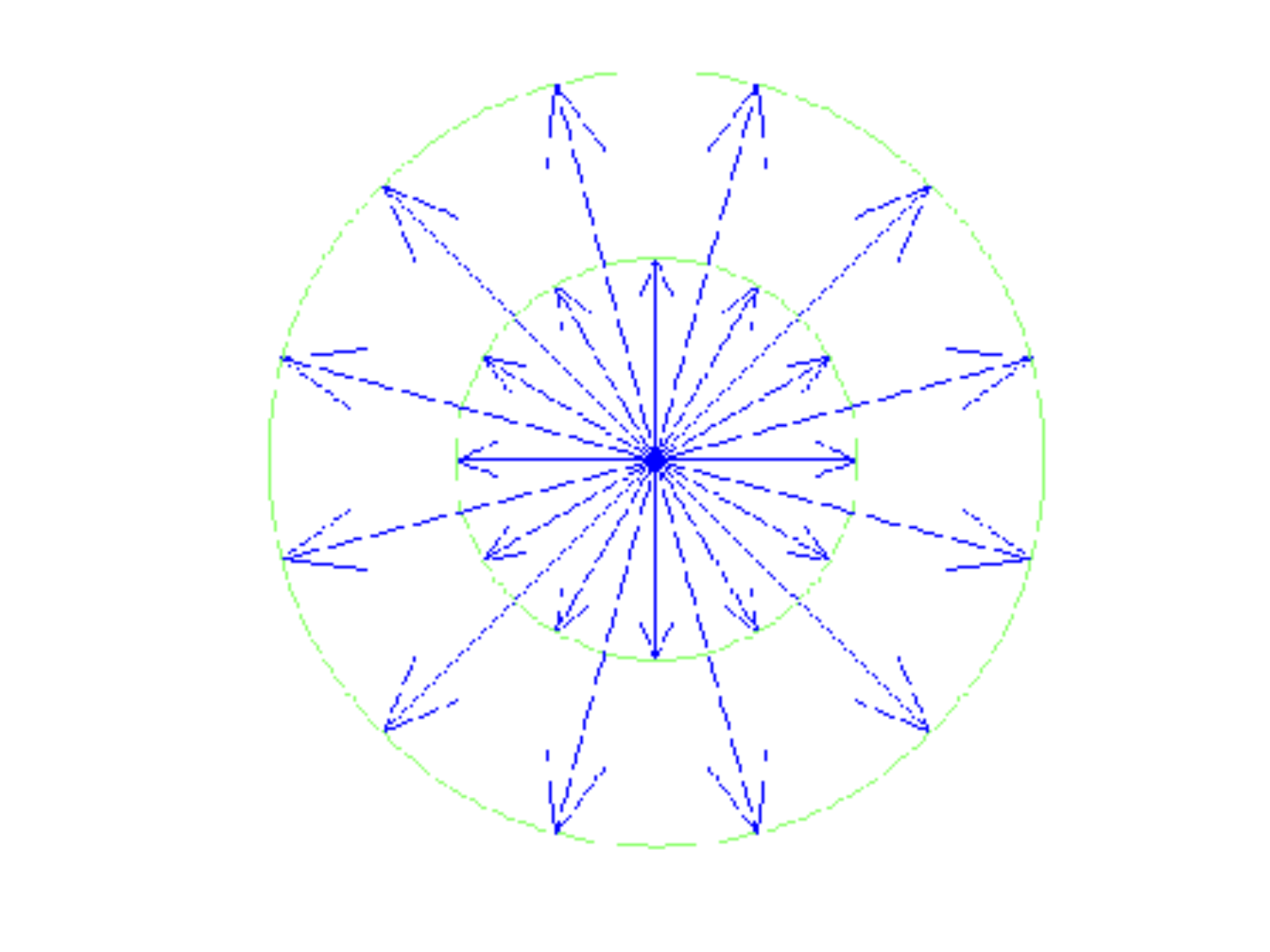}
	\caption{\label{fig:12fold:initial} Initial reciprocal lattice vectors
	for DDQC.}
\end{figure}
\begin{table*}[!htbp]
\caption{\label{tab:DDQC:initial2d} Initial reciprocal lattice vectors
for DDQC in $2$-dimensional space with $q=2\cos(\pi/12)$.}
\centering
\begin{tabular}{|c|c|}
\hline
$ |\mathbf{k}| = 1$ & ($\cos(j\pi/6)$,$\sin(j\pi/6)$), $j =
0,1,\dots, 11$
	\\ \hline
$ |\mathbf{k}| = q$ &($q\cos(j\pi/6+\pi/12)$, $q\sin(j\pi/6+\pi/12)$), $j =
0,1,\dots, 11$
\\\hline
\end{tabular}
\end{table*}
The specific reciprocal vectors are given in Tab.\,\ref{tab:DDQC:initial2d}
when $\mathbf{e}^*_1 = (1,0)$ and $\mathbf{e}^*_2=(0,1)$, as
shown in Fig\,\ref{fig:pamscheme}, are chosen as the basic
reciprocal vectors in 2 dimensions. The distinct nonzero
represented coefficients of these initial reciprocal vectors for
DDQC phase are $1$, $1/2$, $\sqrt{3}/2$,
$2\cos(\pi/12)$, $\cos(\pi/12)$ and $\sqrt{3}\cos(\pi/12)$.
The condition of SDA (\ref{eqn:sda}) must be satisfied when one uses CAM.
Fig.\,\ref{fig:SDA} shows the trend of the approximate error of
SDA, $E_{SDA}$, as the integer $L$ increases. 
\begin{figure}[!htp]
	\centering
		\includegraphics[scale=0.7]{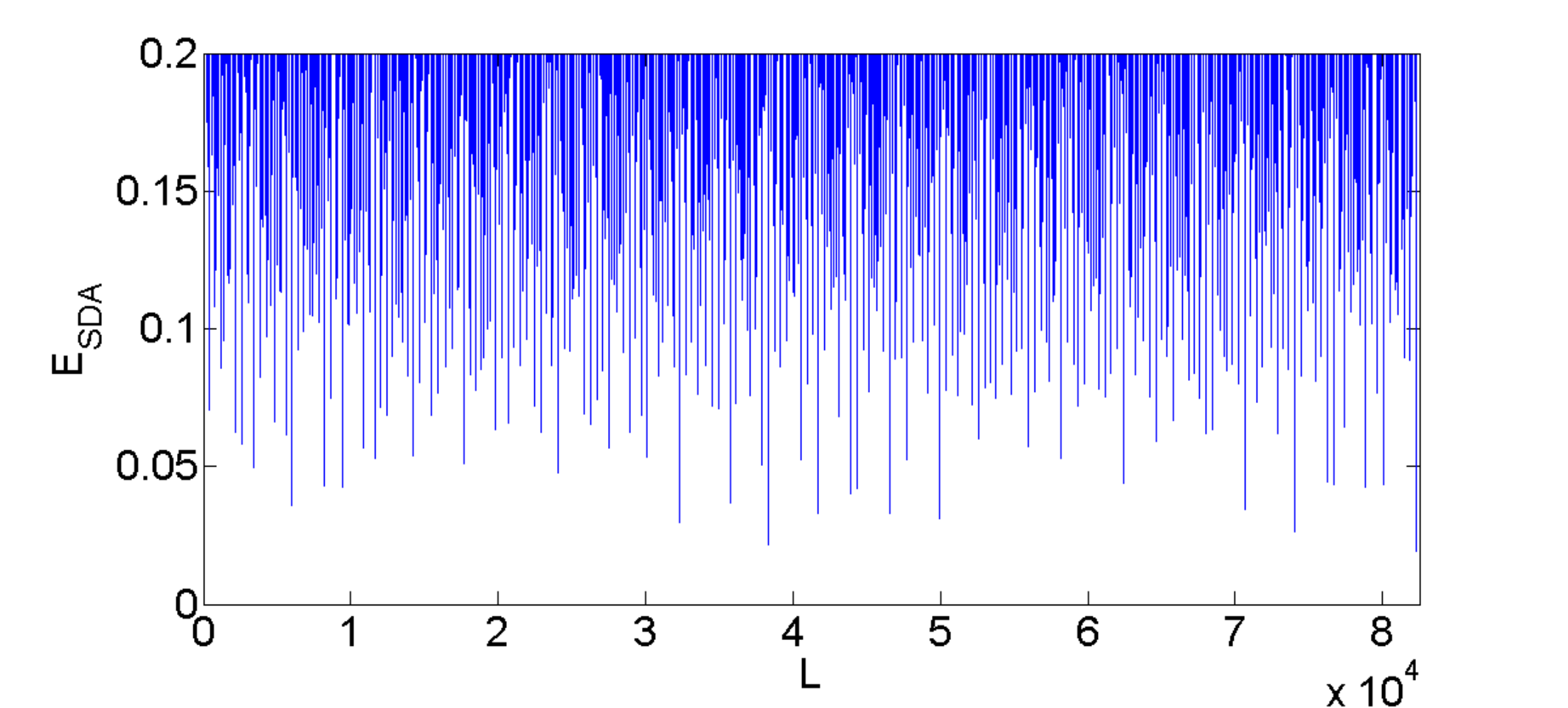}
	\caption{\label{fig:SDA} 
	The trend of the approximate error of SDA, $E_{SDA}$,
	as a function of the integer $L$.
	}
\end{figure}
The convergence of $E_{SDA}$ is not uniform.
\begin{table*}[!htbp]
\caption{\label{tab:DDQC:SDA} 
The minimal integer $L$ for desired approximate error of SDA, $E_{SDA}$.
}
\centering
\begin{tabular}{|c|c|c|c|c|c|c|c|}
\hline
$E_{SDA}$ & 0.19098 & 0.17486 & 0.07042 & 0.04953 & 0.03583 & 0.02961 &
0.01936
\\\hline
$L$ & 30  & 208& 410 &  3404 & 6016 & 32312 &  82262
\\\hline
\end{tabular}
\\
The size of computational box used in CAM is $D_\mathbf{k}=2\pi L$. 
\end{table*}
Tab.\,\ref{tab:DDQC:SDA} gives the minimal integers for desired
approximate error $E_{SDA}$, and 
the corresponding size $D_\mathbf{k}$ of computational box used in CAM.
Then the initial reciprocal vectors in CAM are
$[L\cdot\mathbf{k}]$. From the table, we find that $L$ increases
quickly as the approximate error becomes small.  
Accordingly, the computational cost will increase
greatly. We also find the $E_{SDA}\approx 0.19098$ 
($L=30$) is the least requirement for computing dodecagonal
symmetric structure, it is consistent with the result in
Ref.\,\cite{lifshitz1997theoretical}.
The initial reciprocal vectors in $4$-dimensional space
representation of PM are shown in
Tab.\,\ref{tab:DDQC:initialPM} when the primitive reciprocal
vectors are $\mathbf{b}^*_1=(1,0,0,0)$,
$\mathbf{b}^*_2=(0,1,0,0)$, $\mathbf{b}^*_3=(0,0,1,0)$, and
$\mathbf{b}^*_4=(0,0,0,1)$, as shown in Fig\,\ref{fig:hdmscheme}.
\begin{table*}[!htbp]
\caption{\label{tab:DDQC:initialPM} Initial reciprocal lattice
vectors $\{\mathbf{H}\}$ for DDQC in PM with $q=2\cos(\pi/12)$.}
\centering
\begin{tabular}{|c|c|}
\hline
$ |\mathcal{S}\cdot\mathbf{H}| =1$ &
\makecell{
	(0     1     0   -1)
	(0    -1     0    1)
	(1     0     0    0)
   (-1     0     0    0)
	(0     1     0    0)
	(0    -1     0    0)
	\\
	(0     0     1    0)
	(0     0    -1    0)
	(0     0     0    1)
	(0     0     0   -1)
   (-1     0     1    0)
	(1     0    -1    0)
	}
	\\ \hline
$ |\mathcal{S}\cdot\mathbf{H}| =q$ &
\makecell{
	(1     1     0   -1)
   (-1    -1     0    1)
	(1     1     0    0)
   (-1    -1     0    0)
	(0     1     1    0)
	(0    -1    -1    0)
	\\
	(0     0     1    1)
	(0     0    -1   -1)
   (-1     0     1    1)
	(1     0    -1   -1)
   (-1    -1     1    1)
	(1     1    -1   -1)
	}
\\\hline
\end{tabular}
\\
The projective matrix $\mathcal{S}$ is given by
Eqn.\,(\ref{eqn:DDQC:projMatrix}).
\end{table*}
We can use the $24$ initial reciprocal vectors in
Tab.\,\ref{tab:DDQC:initialPM} to estimate the unit cell of DDQC
in $4$-dimensional space.
From the estimation formula\,(\ref{eqn:box}) and the dual
relationship (\ref{eqn:dualRelation}), 
the size of the unit cell in $n$-dimensional direct space 
is $D_{\mathbf{H}}=2\pi$.
The projective matrix $\mathcal{S}$ of PM is given by
Eqn.\,(\ref{eqn:DDQC:projMatrix}).

In the following, we will compare the computational complexity
between CAM of different unit cell with
$D_{\mathbf{k}}=2\pi\cdot L$ and PM as
discussed in Sec.\,\ref{subsec:complexity}.
In each time step, for the dodecagonal symmetric structure, 
the computational complexity of CAM is 
\begin{align}
	C_{\mathbf{k}} = O\left(2N_{\mathbf{k}}^2 \log
	N_{\mathbf{k}}\right) = O\left(2\left(\frac{2\pi L}{\Delta
	x}\right)^2 \log \left(\frac{2\pi L}{\Delta x}\right)\right),
	\label{eqn:complexity:ddqcCAM}
\end{align}
and the computational complexity of PM is
\begin{align}
	C_{\mathbf{H}} = O\left(4 N_{\mathbf{H}}^4 \log
	N_{\mathbf{H}}\right) = O\left(4 \left(\frac{2\pi}{\Delta
	x}\right)^4 \log \left(\frac{2\pi}{\Delta x}\right) \right).
	\label{eqn:complexity:ddqcPM}
\end{align}
For $L=30$ with $E_{SDA}\approx 0.19098$ in CAM, the
computational complexity $C_{\mathbf{k}}<C_{\mathbf{H}}$ when
$\Delta x$ is smaller than $0.20943951$.  
However, for $L=208$,
$C_{\mathbf{k}}<C_{\mathbf{H}}$ only if $\Delta x > 500.6549$,
which is much larger than the size of the unit cell $D_{\mathbf{H}}$.
It can not be implemented numerically. 
The computational complexity $C_{\mathbf{H}}$ of PM 
is always smaller than that of $C_{\mathbf{k}}$.
For higher approximate accuracy of SDA with larger $L$,
$C_{\mathbf{H}}<C_{\mathbf{k}}$, and total computational
complexity of PM may be less than CAM.

Subsequently we compare the total computational complexity including time
iterations of CAM (with $L=30$) and PM through numerical experiments.
All the methods considered here have been implemented in C
language. Fourier transforms are computed using the
FFTW\,\cite{frigo1998fftw} package. The codes were run in the same
workstation, a Intel(R) Xeon(R) CPU E5450 @3.00GH memory
16 G under linux. The time step size, $\Delta t$, is always
selected as $0.1$ for both methods. To measure error we use
$l^{\infty}$ norm,
\begin{align}
	E_{CAM} = \max_\mathbf{k}\left\{ \left(\frac{\delta F}{\delta
	\phi}\right)(\mathbf{k})\right\}
\end{align}
for CAM, and 
\begin{align}
	E_{PM} = \max_\mathbf{H}\left\{ \left(\frac{\delta F}{\delta
	\phi}\right)(\mathbf{H})\right\}
\end{align}
for PM. 

\begin{figure}[!htp]
	\centering
		\includegraphics[scale=0.6]{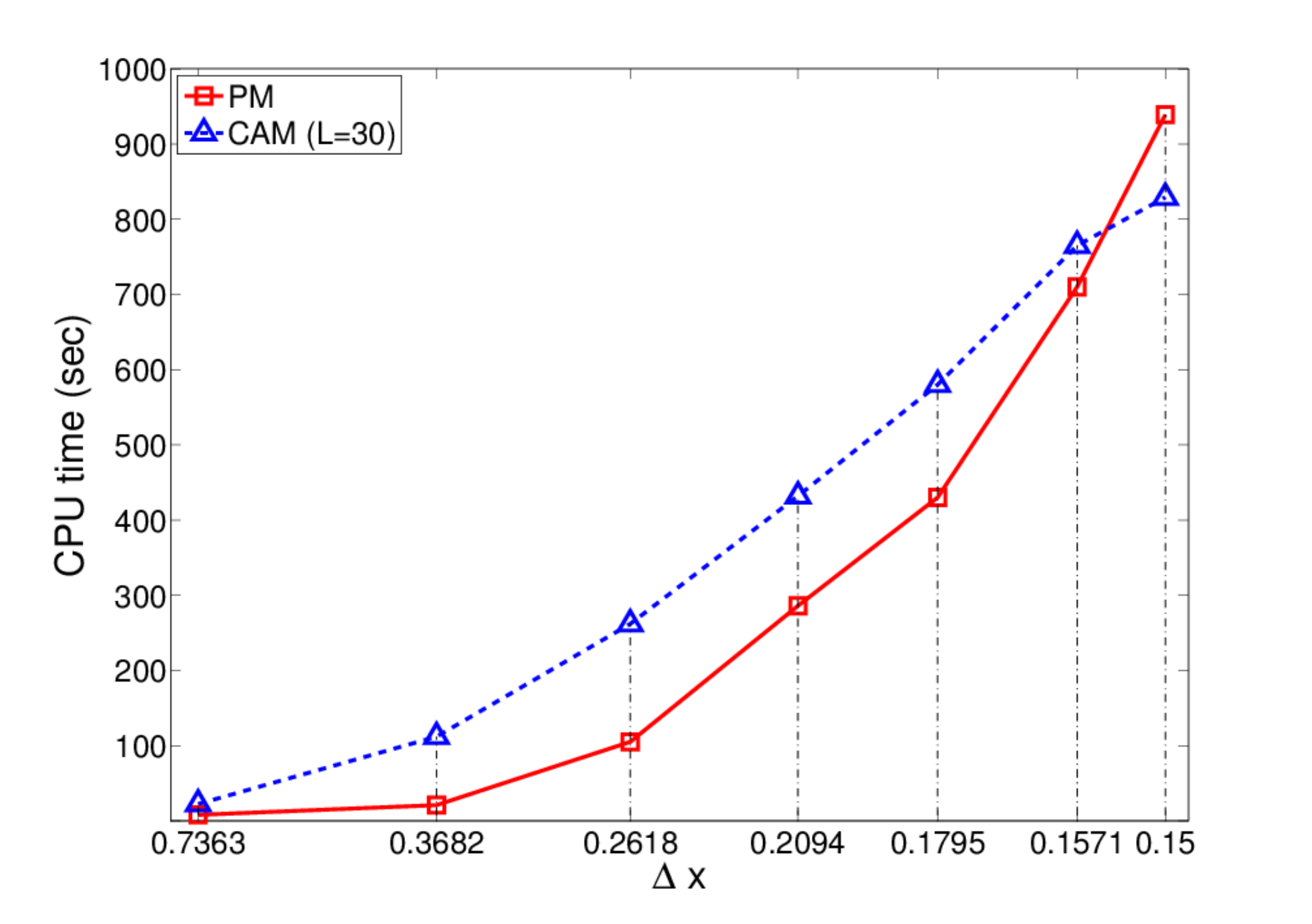}
	\caption{ 
The total CPU time required by CAM ($L=30$) and PM at the error
of $10^{-8}$, with different mesh step size $\Delta x$ at a set
of parameters: $c=50$, $\varepsilon=0.015$, $\alpha=1.0$.
}
\label{fig:DDQC:complexity}
\end{figure}
Fig.\,\ref{fig:DDQC:complexity} shows the total CPU time
reaching an error of $10^{-8}$ by CAM ($L=30$) and PM 
at a set of parameters, $c=50$, $\varepsilon=0.015$,
$\alpha=1.0$, with different discrete resolution. The number of
iterations on different grids is same. 
For CAM method, $N_{t,\mathbf{k}}=848$, while for PM approach,
$N_{t,\mathbf{k}}=472$. It may be that the accuracy of
spacial discretization is enough for the calculated structure at
the set of parameter. For the same reason, the free energy
density calculated by PM is nearly the same on these grids,
$F_{PM}=-3.524067379\mathrm{e}-03$, while the CAM obtains the equal free
energy density $F_{CAM}=-3.457837609\mathrm{e}-03$. 
However, the iterations of CAM, $N_{t,\mathbf{k}}$ is alway large than that of PM,
$N_{t,\mathbf{H}}$. As discussed in
Sec.\,\ref{subsec:result:symmetry}, it may be that PM keeps the
rotational symmetry, while CAM dose not. Therefore PM can reach the
same accuracy more quickly.
The free energy density computed by PM is lower than the CAM. It
may result from the approximate error of SDA in CAM (also see
Fig\,\ref{subfig:penaltyC:energy}).
As Fig\,\ref{fig:DDQC:complexity} demonstrates, the cost of CPU time of PM is
lower than that of CAM until the mesh step size $\Delta x$ is less than
about $0.15$. It is different from the above analysis about the
computational complexity in each time step because of the
different iterations between PM and CAM. 

From these results, we find that with relative large mesh step
size $\Delta x$, PM can obtain enough numerical accuracy with
less computational cost than CAM ($L=30$). CAM always has
approximant error unless $D_{\mathbf{k}}\rightarrow\infty$. 
And as discussed in the section, higher accuracy of approximant
error will result in heavy computational burder. Therefore,
PM has less computational complexity than CAM to high
accuracy in computing these quasicrystals represented in $4$
dimensions, such as dodecagonal symmetric structures.

\subsection{Dodecagonal Symmetric Structure}
\label{subsec:result:symmetry}

The initial reciprocal coefficients, as mentioned in
Sec.\,\ref{subsec:result:complexity}, at which the Fourier
coefficients are nonzero are used as initial values to find
equilibrium dodecagonal symmetric structures.
The simulations of CAM are performed on a $256\times 256$ grids
with $L=30$, and the simulations of PM are performed with
$N_{\mathbf{H}}^4=32^4$ plane waves.
The projective matrix $\mathcal{S}$ of PM is given by
expression\,(\ref{eqn:DDQC:projMatrix}).
We also find that the morphologies and the free energy
density will not change with denser grids. 
The morphology of the dodecagonal approximant computed by CAM
in physical space is given in Fig.\,\ref{subfig:density:cam}.
It is a periodic structure.
The DDQC calculated by PM is shown in
Fig.\,\ref{subfig:density:pm}.
\begin{figure}[!htp]
	\centering
	\subfigure[Dodecagonal crystalline approximant]
	  {
	  \label{subfig:density:cam}
		\includegraphics[scale=0.5]{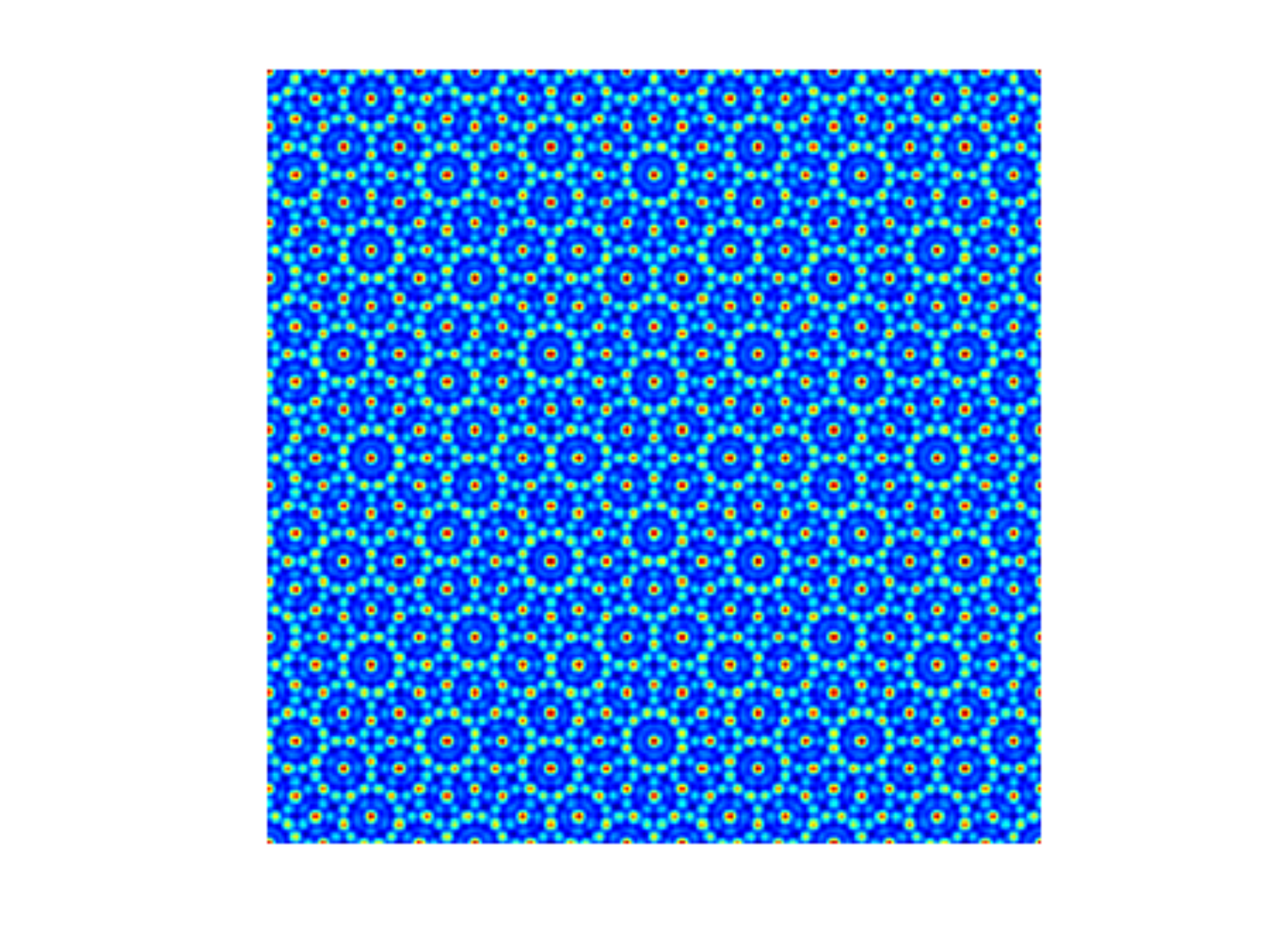}
	  }
	\subfigure[DDQC]
	  {
	  \label{subfig:density:pm}
		\includegraphics[scale=0.5]{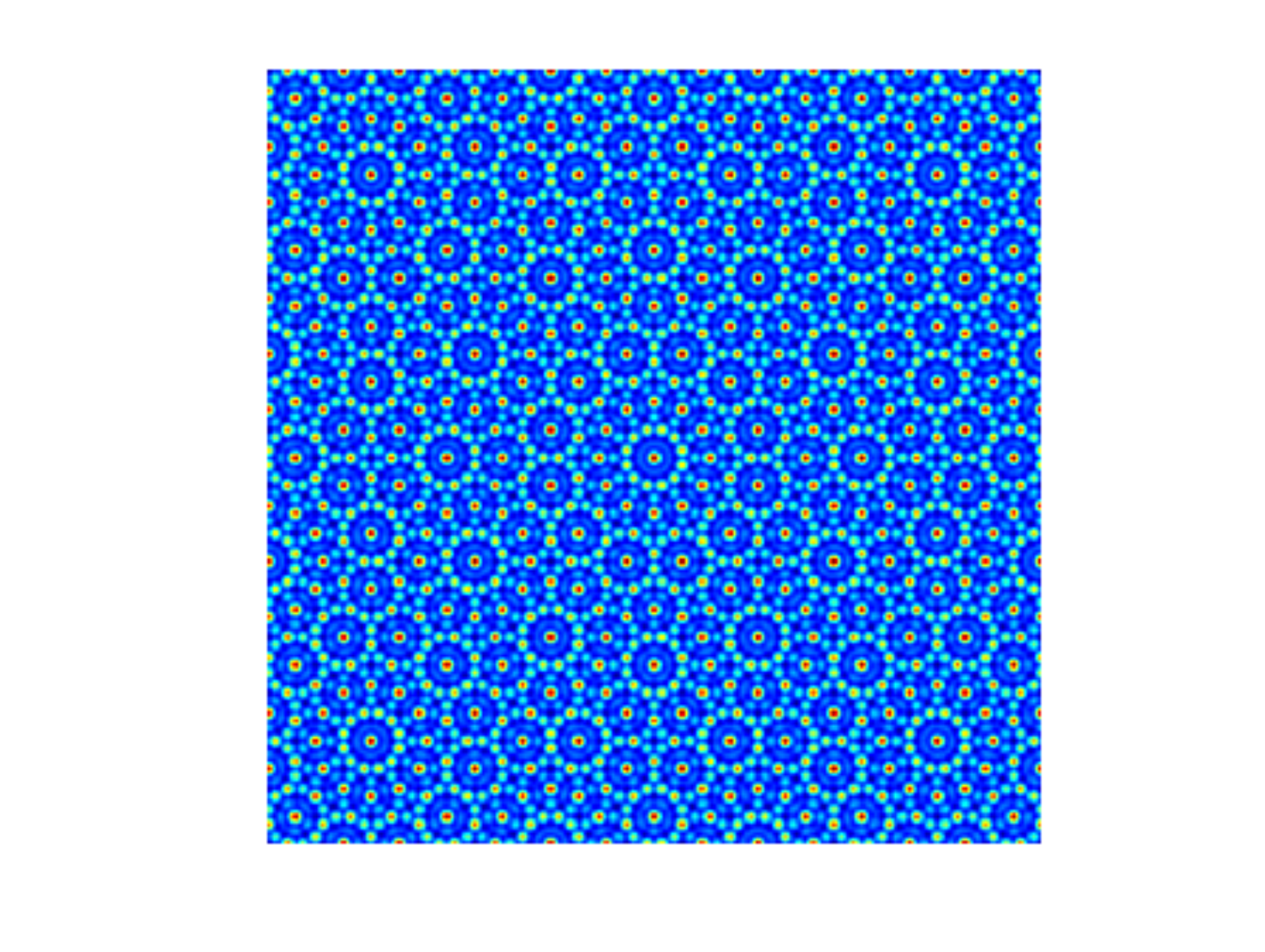}
	  }
	\caption{\label{fig:density} The morphologies of (a).
	Dodecagonal approximant computed by CAM and (b). DDQC
	calculated by PM 
	at $\varepsilon = 0.015$, $\alpha=1$, $c=50$, $q=2\cos(\pi/12)$.}
\end{figure}
Both morphologies demonstrate the 12-fold orientational symmetry
in the physical space, at least locally.
However, further analysis about the relationship of the Fourier
coefficients will come to a different conclusion.
From the rotational symmetry, 
each $12$ reciprocal vectors with dodecagonal rotational symmetry 
on the ring $|\mathbf{k}|=1$
($|\mathcal{S}\cdot\mathbf{H}|=1$), or
$|\mathbf{k}|=q$ ($|\mathcal{S}\cdot\mathbf{H}|=q$), 
should have equal Fourier coefficients.
Tab.\,\ref{tab:DDQC:PAM:coeffs} shows the 
principle Fourier coefficients calculated by CAM. 
Tab.\,\ref{tab:DDQC:PM:coeffs} gives these
Fourier coefficients on principle reciprocal vectors calculated
by PM. From these results, we can find the $12$ Fourier coefficients
calculated by CAM on each characteristic ring are not equal.
\begin{table}[!htbp]
\caption{\label{tab:DDQC:PAM:coeffs}
The Fourier coefficients of the principle reciprocal vectors for
dodecagonal crystalline approximant computed by CAM with
$L=30$ at $\varepsilon = 0.015$, $\alpha=1$, $c=150$, $q=2\cos(\pi/12)$.}
\centering
\begin{tabular}{|c|c|c|}
\hline
$ |\mathbf{k}| = L$ &
\makecell{
	(30,0)   
   (-30,0)	
	 (0,30)	
	 (0,-30)	
	 \\ \hline 
	(26,15)   
   (-26,-15)
	(15,26)
   (-15,-26)
   \\
   (-15,26)  
	(15,-26)
   (-26,15)
   (26,-15)
   }
   & 
\makecell{
	6.106618210e-02
	 \\ \hline 
	5.928204525e-02
	\\
	-
   }
	\\ \hline
$ |\mathbf{k}| = [L\cdot q]$ &
\makecell{
	(41,41)
   (-41,-41)
   (-41,41)
   (41,-41)
	 \\ \hline 
	(56,15)
   (-56,-15)
	(15,56)
   (-15,-56)
	 \\ 
   (-15,56)	
	(15,-56)
   (-56,15)
	(56,-15)
	}
	&
\makecell{
	5.458000478e-02
	 \\ \hline 
	5.683018766e-02
	  \\
	  -
	}
	\\ \hline
\end{tabular}
\end{table}
In contrast, as Tab.\,\ref{tab:DDQC:PM:coeffs} shown, 
PM can obtain the $12$ equal Fourier coefficients on each ring.
We also find the same phenomenon at different parameter coordinates.
\begin{table}[!htbp]
\caption{\label{tab:DDQC:PM:coeffs}
The Fourier coefficients of the principle reciprocal vectors for DDQC
computed by PM at 
$\varepsilon = 0.015$, $\alpha=1$, $c=150$, $q=2\cos(\pi/12)$.}
\centering
\begin{tabular}{|c|c|c|}
\hline
$ |\mathcal{S}\cdot\mathbf{H}| =1$ &
\makecell{
	(0     1     0   -1)
	(0    -1     0    1)
	(1     0     0    0)
   (-1     0     0    0)
   \\
	(0     1     0    0)
	(0    -1     0    0)
	(0     0     1    0)
	(0     0    -1    0)
	\\
	(0     0     0    1)
	(0     0     0   -1)
   (-1     0     1    0)
	(1     0    -1    0)
	}
	&
	5.856822141e-02
	\\ \hline
$ |\mathcal{S}\cdot\mathbf{H}| =q$ &
\makecell{
	(1     1     0   -1)
   (-1    -1     0    1)
	(1     1     0    0)
   (-1    -1     0    0)
	\\
	(0     1     1    0)
	(0    -1    -1    0)
	(0     0     1    1)
	(0     0    -1   -1)
	\\
   (-1     0     1    1)
	(1     0    -1   -1)
   (-1    -1     1    1)
	(1     1    -1   -1)
	}
	&
	5.855442187e-02
\\\hline
\end{tabular}
\end{table}
For CAM, we also use higher approximate accuracy $E_{SDA}\approx0.07042
$, with the computational box $D_{\mathbf{k}}\times
D_{\mathbf{k}}$, $D_{\mathbf{k}}=820\pi$. The numerical
experiments are implemented
on a $4096\times4096$ grids. We find that CAM can capture the
dodecagonal approximant, however, it still can not obtain 12 equal
Fourier coefficients on each ring
as well as the above numerical experiments.
Therefore, we come to a conclusion that PM can keep
the non-crystallographic symmetry to high accuracy, 
while CAM can not unless $D_{\mathbf{k}}\rightarrow\infty$.

\subsection{Free Energy Density}
\label{subsec:energy}

In this section, we compare the free energy density computed by
CAM, PM and single-wave approximation (SWA) approaches. 
The SWA method uses principal Fourier vectors to calculate the
free energy functional analytically under some
constraints\,\cite{chaikin1995principles}. It translates the
energy functional into a single-variable or multi-variable
function. The minimum of free energy can be approximately
obtained by minimizing the reduced energy function with
respect to these variables. In the Lifshitz-Petrich model, 
the principal Fourier vectors of DDQC have been shown in
Fig\,\ref{fig:12fold:initial}. When $c\rightarrow\infty$, the
reduced free energy function of DDQC pattern computed by
SWA\,\cite{lifshitz1997theoretical, jiang2013theory} is
\begin{align}
	F_{12} = 99(\phi_1^4+\phi_q^4) + 144(\phi_1^3+\phi_1^3\phi_q) +
	360\phi_1^2\phi_q^2 
	-24(\phi_1^2\phi_q+\phi_1^2 \phi_q) -8(\phi_1^3+\phi_q^3) -
	6\varepsilon(\phi_1^2+\phi_q^2)/\alpha^2,
	\label{}
\end{align}
$\phi_1, \phi_q\in\mathbb{R}$ stand for Fourier coefficients on the
$|\mathbf{k}|=1$ and $|\mathbf{k}|=q$ rings, respectively.
The approximated minimum of DDQC can be obtained by
minimizing the above equation with respect to $\phi_1$ and $\phi_q$.
\begin{figure}[!htp]
	\centering
	\subfigure[Free energy density]
	  {
		\label{subfig:penaltyC:energy}
		\includegraphics[scale=0.36]{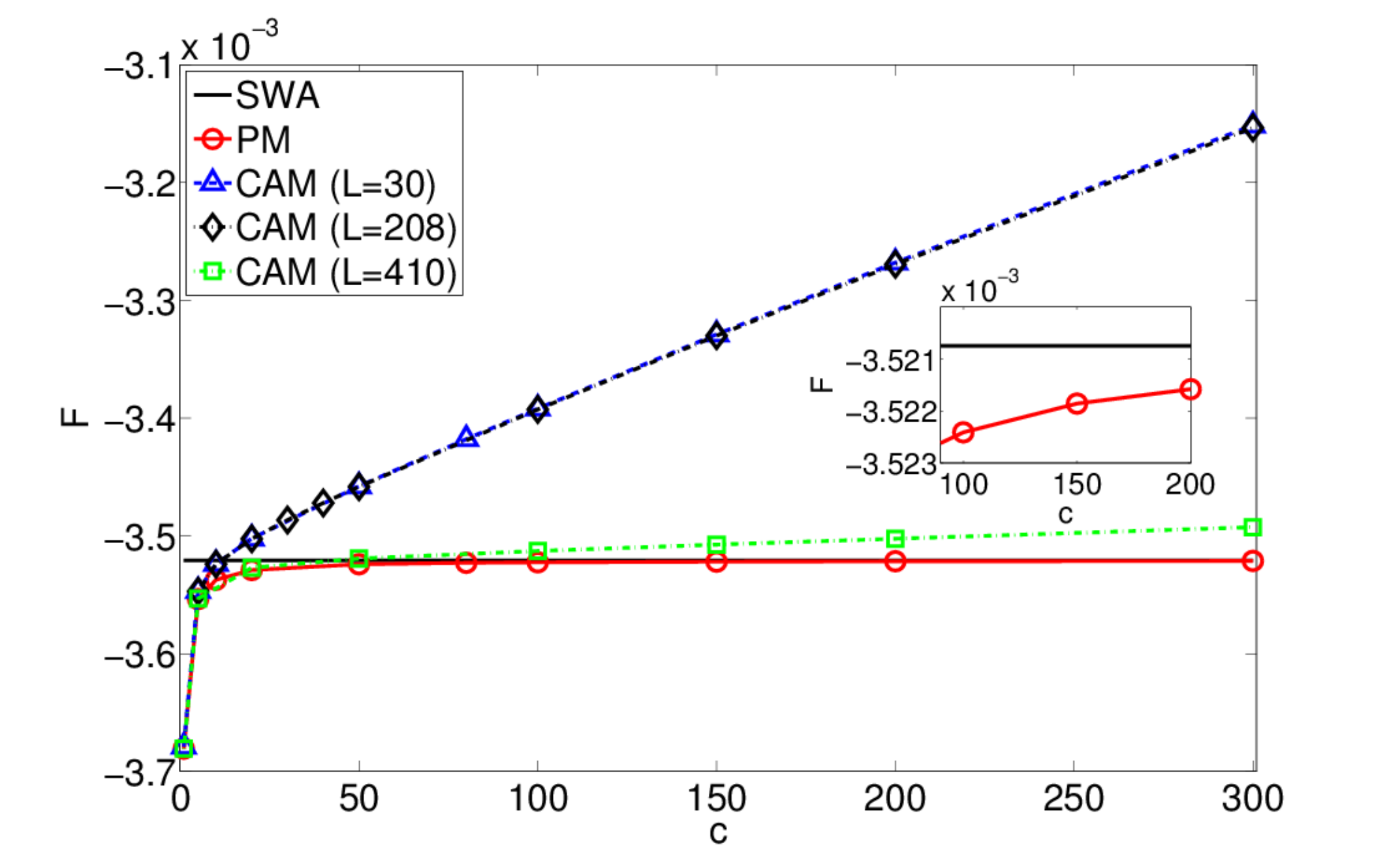}
	  }
	  \subfigure[Interfacial energy (Laplacian terms)]
	  {
		\label{subfig:penaltyC:laplace}
		\includegraphics[scale=0.36]{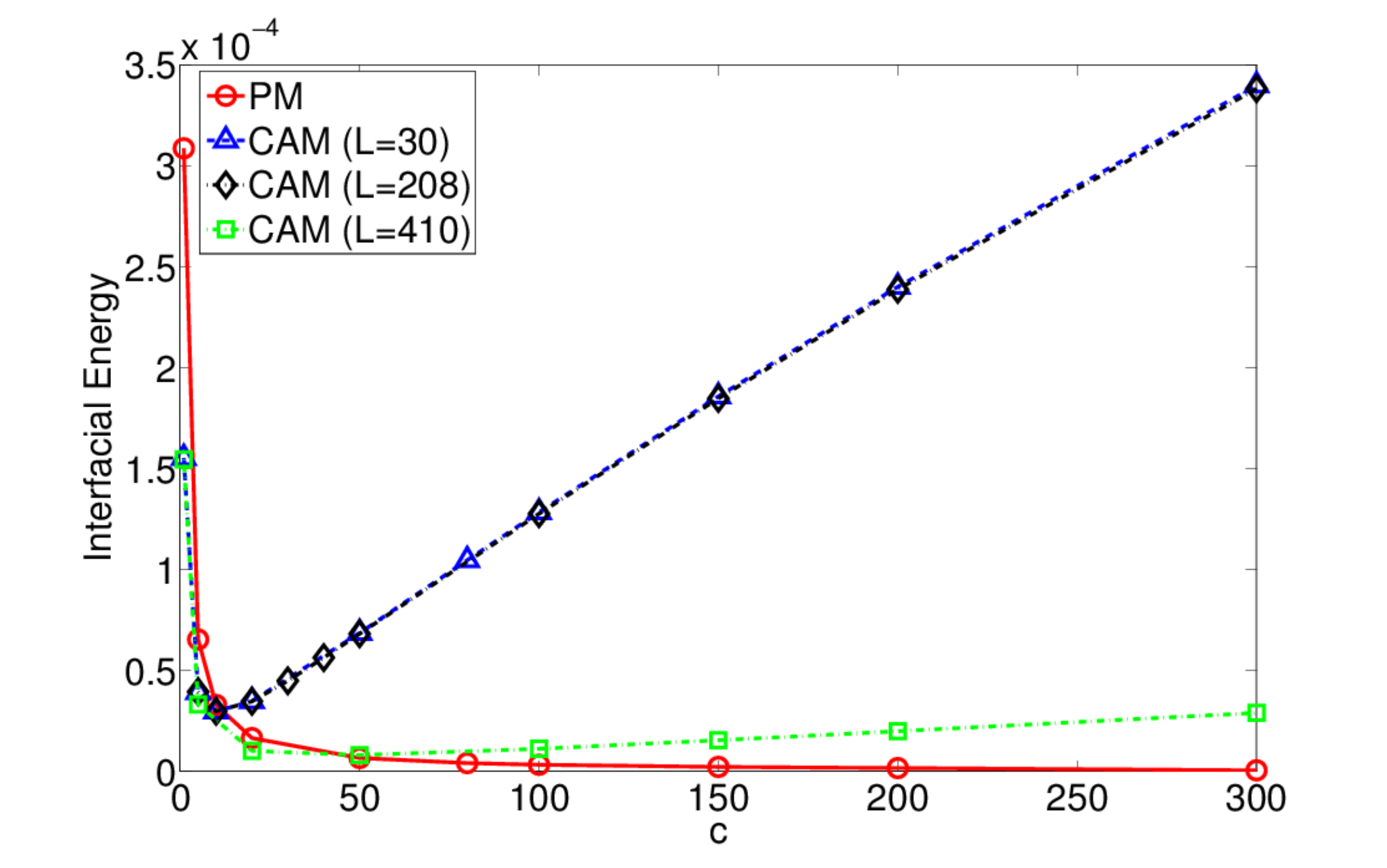}
	  }
	  \caption{(a). Free energy density calculated by PM and CAM,
	  as a function of the penalty factor $c$,
	  relative to that of SWA
	  ($c\rightarrow\infty$), with $q=2\cos(\pi/12)$, $\alpha=1.0$,
	  $\varepsilon = 0.015$. (b). Corresponding interfacial energy
	  computed by PM and CAM.
	  }
	\label{fig:penaltyC}
\end{figure}

\begin{table*}[!htbp]
	\caption{\label{tab:cpu:c}
	The CPU time required for minimizing the free energy
	\eqref{eqn:LP} by the PM and CAM approaches at error of
	$10^{-6}$ with $q=2\cos(\pi/12)$, $\alpha=1.0$ $\varepsilon = 0.015$
	for different $c$. In simulations, $\Delta t = 0.1$; 
	$N_{\bm H}=24$ utilized in PM;
	$N_{\bm k}=720, 2048, 4096$, used in CAM corresponding to $L=30,
	208, 410$, respectively.  }
\centering
\begin{tabular}{|c||c|c|c|c|c|c|c|}
\hline
 & \multicolumn{7}{c|}{CPU time (sec)}
\\ \hline\hline
$c$ & 5 & 20 & 50 & 100 & 150 & 200 & 300
\\ \hline 
PM	& 45.86 & 40.91 & 40.38 & 39.49& 39.69 & 39.72 & 39.39 
\\ \hline
CAM ($L=30$) & 122.49 & 136.97 & 143.10 & 181.89 & 189.29 & 185.33 & 197.05
\\ \hline
CAM ($L=208$) & 1172.56 & 1381.57& 1567.82 & 1716.40 &1764.55 & 1896.65 & 1847.07
\\ \hline
CAM ($L=410$) & 3457.19 &  3655.40 &3397.97 & 3803.35 & 4035.10& 4471.93& 4767.31 
\\ \hline
\end{tabular}
\end{table*}

Fig.\,\ref{subfig:penaltyC:energy} shows the free energy density
calculated by PM and CAM ($L=30, 208, 410$), as a function of the
penalty factor $c$, relative to that of SWA
($c\rightarrow\infty$) with $q=2\cos(\pi/12)$, $\alpha=1.0$,
$\varepsilon = 0.015$. 
In numerical simulations, CAM are performed with $720\times 720$
plan waves ($\Delta x = 0.2618$) for $L=30$; 
with $2048\times2048$ plan waves for $L=208$;
with $4096\times4096$ plan waves for $L=410$. 
$32^4$ ($\Delta x = 0.2618$) plan waves were utilized in PM approach.
The free energy density computed by CAM is
heavily dependent on the approximate error of SDA. 
As Tab.\,\ref{tab:DDQC:SDA} shown, when $L=30$ and $208$, the
approximate errors of SDA are $0.19098$ and $0.17486$. Under the
approximate error, the CAM manifests the nearly same behavior.
The free energy density $F_{CAM}$ computed by CAM is less than 
$F_{SWA}$ obtained by SWA until the penalty factor $c$ increases
to about $10$. Then $F_{CAM}$ is larger than $F_{SWA}$ as $c$ increases.
The reason is that the CAM can not
control the principal reciprocal vectors located on $|\mathbf{k}|=1$
and $|\mathbf{k}|=q$ when the penalty factor $c$ increases, as
shown in Fig.\,\ref{subfig:penaltyC:laplace}.
In other words, CAM can not calculate the interfacial energy
(Laplacian terms) accurately when $c$ is large. However the
role of the differential terms in the Lifshitz-Petrich model is
to keep the interactions at two characteristic scales. It is one of 
the reasons that the Lifshitz-Petrich model is able to stabilize
quasicrystals\,\cite{lifshitz1997theoretical}.
Correspondingly, the $F_{CAM}$ diverges from the true free
energy density. We have also used $L=410$ which improves the approximate
accuracy of SDA to $E_{SDA}= 0.07042$ in CAM to observe the free
energy density. It improves the accuracy of free energy, however, 
the same problem appeared when $c$ is larger than $50$.
In contrast, the free energy density $F_{PM}$ calculated by PM is always
smaller than $F_{SWA}$ for all $c$ and converges to the $F_{SWA}$
as $c\rightarrow\infty$. It is because the more basic functions
are used in simulations by PM than that of SWA. 
The approximated space of PM is more precise than that of SWA.
The PM can also maintain its
principle reciprocal vectors located on scales $1$ and $q$ which is
consistent with the model, as shown in
Fig.\,\ref{subfig:penaltyC:laplace}. 
The corresponding CPU times spent in simulations are shown in  
Tab.\,\ref{tab:cpu:c}. Compared with the CAM method,
the PM can evaluate the free energy density to high numerical
accuracy with economical computational cost in dodecagonal
symmetric phase.

\section{Conclusions and Outlook}
\label{sec:conclusion}
In the article, we summary the features of CAM approach, 
and point out the advantage of using the periodic condition and
the restriction condition of the SDA in this method. 
The errors of CAM come from two parts: the approximate error of
SDA and numerical discretization.
Subsequently we propose a new numerical method, the PM,
with the help of higher-dimensional
reciprocal space. The developed method overcomes the restriction
of SDA and enables us view a quasicrystal as a
periodic structure, and uses the periodic condition conveniently. 
The projection method can reduce the computational effort
efficiently by computed quasicrystals in a higher-dimensional
unit cell and using the pseudospectral method.  
By applying two methods to the Lifshitz-Petrich model, 
we analyze the computational complexity. The
computational complexity of CAM is
dependent on the approximate error and the numerical resolution,
while that of PM is only dependent upon the numerical resolution. 
Specially, in computing dodecagonal symmetric pattern, 
PM has less computational complexity than that of CAM.
We also find that our approach can keep the rotational symmetry 
accurately, and more significantly, the present algorithm can calculate
the free energy density to high accuracy without any approximate error of SDA.
However, the dimensions of our computational space are dependent
on the rotational symmetry of quasicrystals. 
For quasiperiodic structures with 5-, 8-, 10- and 12-fold symmetry,
the computational space dimension is four. PM has accepted
computational complexity. For a quasicrystal with of 7-, 9-, and
18-fold quasicrystals, the dimension of computational space is up
to six. It is also emphasized that the recently discovered
quasicrystals can be all represented as periodic structures in less
than $6$-dimensional space. Our future work will focus
on the improvement of the computational efficiency of the
projection method, and make it to $6$-dimensional cases. 
Finally, we should point out that 
PM can be also applied to study general $d$-dimensional aperiodic
structures\,\cite{janssen2007aperiodic} whose Fourier spectrum
consists of $\delta$-peaks, $\mathbf{H}=\sum_{i=1}^n h_i\mathbf{b}^*_i$,
$h_i\in\mathbb{Z}$, of rank $n$ ($n>d$) with basis vectors
$\mathbf{b}_i^*$, $i=1,2,\dots,n$.

\section*{Acknowledgments}
The authors would like to acknowledge Dr.\,Robert A. Wickham for
his help with the English grammar and useful advice.
The work is supported by the National 
Science Foundation of China 21274005 and 50930003, and the 
China Postdoctoral Science Foundation 2011M500179.

\end{document}